\documentclass[12pt]{article}

\usepackage{draft, graphics, amsmath, amsthm, subfigure, float, hyperref}

\newtheorem{prop}{Proposition}

\newtheorem{cor}{Lemma}

\begin{document}

\begin{titlepage}

\begin{center}

\hfill \\
\hfill \\
\vskip 1cm

\title{Bootstrapping 2D CFTs in the Semiclassical Limit}

\author{Chi-Ming Chang$^\bullet$, Ying-Hsuan Lin$^\circ$}

\address{$^\bullet$Center for Theoretical Physics and Department of Physics,  \\
University of California, Berkeley, CA 94704 USA}
\address{$^\circ$Jefferson Physical Laboratory, Harvard University, \\
Cambridge, MA 02138 USA}

\email{cmchang@berkeley.edu, yhlin@physics.harvard.edu}

\end{center}

\abstract{We study two-dimensional conformal field theories in the semiclassical limit.  In this limit, the four-point function is dominated by intermediate primaries of particular weights along with their descendants, and the crossing equations simplify drastically.  For a four-point function receiving sufficiently small contributions from the light primaries, the structure constants involving heavy primaries follow a universal formula.  Applying our results to the four-point function of the $\bZ_2$ twist field in the symmetric product orbifold, we produce the Hellerman bound and the logarithmically corrected Cardy formula that is valid for $h \geq c/12$.}

\vfill

\end{titlepage}

\eject

\tableofcontents

\section{Introduction}

The defining data of a conformal field theory (CFT) consist of a set of distinguished operators (primaries) and a set of structure constants.  On top of these are the requirements of crossing symmetry, that the operator product expansion is associative.  Once these constraints are satisfied, these data together with conformal symmetry generate all the correlation functions in the theory, as sums or integrals of conformal blocks weighted by the structure constants.  The conformal bootstrap program, initiated in \cite{Polyakov:1974gs,Ferrara:1973yt,Mack:1975jr}, aims to solve and classify conformal field theories by analyzing these constraints.

While the conformal bootstrap program achieved huge success in solving and classifying rational conformal field theories in two dimensions \cite{Belavin:1984vu,Knizhnik:1984nr,Gepner:1986wi,Bouwknegt:1992wg,Verlinde:1988sn,Dijkgraaf:1988tf,Moore:1988uz,Moore:1988ss}, outside this ``tamed zoo'', the crossing equations are an infinite set of equations depending on infinitely many variables, and a systematic solution is not known.  Recent developments based on numerical methods allow the extraction of certain information, such as bounds on the gap in operator product expansions and bounds on the central charge \cite{Rattazzi:2008pe,ElShowk:2012ht}.  A more analytic approach to bootstrap is to consider limits in which the crossing equations simplify \cite{Fitzpatrick:2012yx,Komargodski:2012ek,Alday:2015eya,Ooguri:2015Talk}. This paper takes the second approach, and study the crossing equations in the semiclassical limit of two-dimensional conformal field theories.

The semiclassical limit is motivated by holography \cite{Maldacena:1997re,Gubser:1998bc,Witten:1998qj}.  Two-dimensional conformal field theories are holographically dual to three-dimensional quantum gravity in asymptotically anti de-Sitter (AdS) space, whose curvature radius (inverse bulk coupling) is equal to the boundary central charge \cite{Brown:1986nw}.  Perturbatively, the bulk spectrum consists of two distinguished classes of states:  a light spectrum containing boundary gravitons and perturbative string states, whose energies do not scale with the central charge, and
a heavy spectrum of non-perturbative states with energies of the order of the AdS curvature, responsible for the microstates of BTZ black holes \cite{Banados:1992wn}.  In order to examine the collective dynamics of the heavy states, the semiclassical limit takes both the central charge and the operator dimensions large while keeping their ratios fixed.  In this limit, the correlation functions of primary operators admit a perturbative expansion in the inverse central charge.  To leading order, the Virasoro block decomposition of the correlation functions are dominated by intermediate primaries of particular weights (a saddle).  As the cross ratio $x$ varies, correlation functions can exhibit ``phase transitions'' due to discontinuous jumps of the weight of the dominant saddle.

This paper studies the four-point function of identical primary operators in the semiclassical limit.  Here we list a summary of our main results.
\begin{enumerate}
\item  At the crossing symmetric point, if there is a single dominant saddle, then its weight is fixed by conformal symmetry.  Away from the crossing symmetric point, the weight of the dominant saddle must be smaller than this fixed value for $x < 1/2$ and larger for $x > 1/2$.  See Proposition~\ref{Prop1}.

\item  If the four-point function receives sufficiently small contributions from light primaries, then the structure constants involving heaving primaries follow a universal formula, {\it \`a la} Proposition~\ref{Prop2b}.

\item  We study the $\bZ_2$ twist field four-point function in the symmetric product orbifold.  Proposition~\ref{Prop3} presents a logarithmically corrected Cardy formula that is valid for $h \geq c/12$.

\end{enumerate}

The sections are organized as follows.  In Appendix \ref{SemiclassicalLimit}, we give a definition of the semiclassical limit.  In Section~\ref{BootstrapSemiclassical}, we study the crossing equation in the semiclassical limit and derive universal constraints.  In Section~\ref{Applications}, we examine specific examples, including Liouville theory, product orbifold theories, and meromorphic CFTs.  In Section~\ref{gravity}, we discuss the gravity dual of classical Virasoro blocks, and implications of our bootstrap results in the gravity context.

\section{Bootstrap in the semiclassical limit}
\label{BootstrapSemiclassical}

The semiclassical limit of a family of two-dimensional conformal field theories is the limit of large central charge $c$ while simultaneously scaling the operator weights with $c$.  See Appendix~\ref{SemiclassicalLimit} for a more careful definition.  In this limit, the crossing equation simplifies drastically, because except for a measure zero set of cross ratios, the sum over intermediate states in either the $s$-channel or the $t$-channel is dominated by just one saddle.  In theories with a gap of order $c$, we will see that the structure constants exhibit certain universal behaviors in the semiclassical limit.

To simplify the discussion, we omit anti-holomorphic variables, but the generalization is straightforward.

\subsection{Review of the conformal bootstrap}
\label{Review of the conformal bootstrap}

Conformal symmetry constrains the four-point function of primary operators to take the form
\ie\vev{\sigma_a(x_1)\sigma_b(x_2)\sigma_c(x_3)\sigma_d(x_4)} &= x_{14}^{-2h_a}x_{24}^{h_a-h_b+h_c-h_d} x_{34}^{h_a+h_b-h_c-h_d}x_{23}^{h_d-h_a-h_b-h_c} F_{abcd}(x),
\fe
where $h_a$ is the conformal weight of $\sigma_a$, etc, and $F_{abcd}(x)$ is a function of the cross ratio $x = { x_{12}x_{34} / x_{14}x_{32} }$.  The four-point function can factorize in different channels.  The $s$-channel
corresponds to the fusion of the primary operators $\sigma_a(x_1)$ and $\sigma_b(x_2)$, and gives an expansion at $x = 0$.  The operators appearing in the operator product expansion (OPE) of $\sigma_a(x_1)$ and $\sigma_b(x_2)$ are organized into representations of the Virasoro algebra.  Each representation contains a primary operator and its descendants.

An inner product on the vector space of primary operators is provided by the two-point function $\langle \sigma_a (0) \sigma_b(1) \rangle$.  We pick an orthonormal basis $\cP$ with respect to this inner product.  Each primary operator and its descendants contribute to the four-point function by a Virasoro block ${\cal F}(h_a,h_b,h_c,h_d, h, c | x)$.  The four-point function can be expanded as
\ie\label{VBD}
F_{abcd}(x) = \sum_{h} C_{\sigma_a\sigma_b}(h) C_{\sigma_c\sigma_d}(h) {\cal F}(h_a,h_b,h_c,h_d, h, c | x).
\fe
A similar expansion exists in the $t$-channel by fusing $\sigma_a$ with $\sigma_d$, and in the $u$-channel by fusing $\sigma_a$ with $\sigma_c$.

In the rest of this paper, we specialize to the four-point function of identical scalar primaries $\vev{\sigma_{ext} \sigma_{ext} \sigma_{ext} \sigma_{ext}}$.  The expansion coefficients $C^2_{\sigma_{ext}\sigma_{ext}}(h)$ are equal to the sum of structure constants squared over all weight-$h$ primaries appearing in the $\sigma_{ext} \times \sigma_{ext}$ OPE, i.e.,
\ie
C^2_{\sigma_{ext}\sigma_{ext}}({h}) \equiv \sum_{\phi \in {\cal P}^{\sigma_{ext}\times\sigma_{ext}}_{h}} C^2_{\sigma_{ext}\sigma_{ext}\phi}, \quad {\cal P}^{\sigma_{ext}\times\sigma_{ext}}_{h}=\{\phi \in {\cal P}^{\sigma_{ext} \times \sigma_{ext}} | h_\phi = h \}.
\fe
They are real and non-negative if we assume unitarity.  Crossing symmetry, or equivalently the associativity of the OPE algebra equates the four-point function expanded in different channels.  We will analyze the crossing equation between the $s$- and $t$-channels
\ie
F(x) \equiv F_{\sigma_{ext} \sigma_{ext} \sigma_{ext} \sigma_{ext}}(x) = F(1-x).
\fe

\subsection{Crossing symmetry in the semiclassical limit}

Given a sequence of CFTs, the semiclassical limit of a four-point function $\vev{\sigma_{ext} \sigma_{ext}\sigma_{ext} \sigma_{ext}}$ is the limit of large central charge $c$ while taking the operator weights $h_{ext}$ to scale with $c$ (fixed $m_{ext} = h_{ext} / c$).  When speaking of correlation functions, in general it is impossible to keep track of a particular primary operator in a sequence of CFTs, so the best we can do is to consider ``{correlation function densities}'' in the semiclassical limit.  See Appendix~\ref{SemiclassicalLimit} for a definition.  We omit these details in this section, and simply refer to them as correlation functions.

It is observed that the Virasoro block admits a semiclassical expansion \cite{Belavin:1984vu,Zamolodchikov:1985ie}
\ie
\label{LCCB}
{\cal F}(h_{ext}, h, c | x) &= \exp\left[ - {c \over 6} f \left( m_{ext}, m | x \right) \right] g\left(m_{ext}, m, c | x \right),
\\
g(m_{ext},m,c|x)  &= \sum_{k = 0}^\infty c^{-k} {g_k(m_{ext}, m | x)}.
\fe
The functions $f$ and $g_k$ can be computed order by order in an $x$-expansion.  The expansions for $f$ and $g_0$ to the first few orders are presented in Appendix~\ref{SemiclassicalCB}.  Our analysis will assume the following numerically observed properties of the semiclassical Virasoro blocks.  For fixed $m_{ext} \leq 1/2$,
\begin{enumerate}
\item  $f'(m_{ext}, m | 1/2)$ is monotonically decreasing in $m$, and crosses zero only once.
\label{NSfp}
\item  $f(m_{ext}, m_2 | x) - f(m_{ext}, m_1 | x)$ is monotonically decreasing in $0 < x < 1$, for arbitrary fixed internal weights $m_2 > m_1\geq 0$.
\label{NSfm2-fm1}
\item  $g_0(m_{ext},m|x) > 0$ for all internal weights $m \geq 0$ and cross ratios $0 \leq x < 1$.
\label{Positiveg}
\end{enumerate}
To use these properties, we will {\bf restrict to $\bf m_{ext} \leq 1/2$}, which is a relatively loose bound compared to either the operators accounting for the microstates of the zero mass BTZ black hole, $m_{BTZ} = 1/24$, or the Hellerman bound \cite{Hellerman:2009bu} on the gap in the spectrum of primaries $m_{gap} \leq 1/12$.  The study of $m_{ext} > 1/2$ is left for future investigation.

In order to satisfy crossing symmetry, the summed structure constants squared which are the coefficients in the Virasoro block decomposition \eqref{VBD} must also admit a semiclassical expansion
\ie\label{SCOPEC}
C^2_{\sigma_{ext}\sigma_{ext}}(m) = \exp\big[ c \,p_{\sigma_{ext}}(m) \big]\big(q_{\sigma_{ext}}(m)+ {\cal O}\left(1/\sqrt{c}\right)\big).
\fe
In theories with a discrete spectrum, the summed structure constants squared is a sum of delta functions.  In the semiclassical limit, this distribution can be approximated by a continuous distribution plus isolated delta functions,
\ie
q_{\sigma_{ext}}(m) = \sum_i q^i_{\sigma_{ext}}\delta(m-m_i) + \sqrt{c}\,q^{cont}_{\sigma_{ext}}(m).
\fe
Here we adopt a normalization such that if the CFT has an order $c$ gap above the vacuum state, then $q_{\sigma_{ext}}^{vac} = 1$.  As we will see, the $\sqrt{c}$ factor in front of the continuous distribution $q^{cont}_{\sigma_{ext}}(m)$ is required for it to be comparable with the delta functions in the large central charge expansion.

For notational simplicity, we define the {\it classical branching ratio} as
\ie
\label{BranchingRatio}
S_{\sigma_{ext}}(m | x) \equiv  p_{\sigma_{ext}}(m)- {1 \over 6} f \left( m_{ext}, m | x \right).
\fe
The crossing equation at large $c$ is
\ie
\label{LCBS}
{\cal O}(1 / c) &= \bigg\{ \sum_{m \in {\cal S}_x} \exp\left[ c \,S_{\sigma_{ext}}\left( m | x \right) \right] q_{\sigma_{ext}}(m) \widetilde g_0\left(m_{ext}, m | x \right) \bigg\} - (x \to 1-x),
\fe
where ${\cal S}_x$ denotes the set of weights that maximize $S_{\sigma_{ext}}(m | x)$ globally, and $\widetilde g_0\left(m_{ext}, m | x \right)$ is defined to also include the one-loop contribution near the saddle point,
\ie
& \widetilde g_0\left(m_{ext}, m | x \right)
\\
&= \begin{cases}
    g_0\left(m_{ext}, m | x \right) &\text{if $m$ is at a delta function,}
    \\
    g_0\left(m_{ext}, m | x \right) \times \sqrt{- {2\pi  \over c \, \partial_m^2 S_{\sigma_{ext}}(m|x)}} &\text{if $m$ is inside the continuum.}
\end{cases}
\fe 
We presently analyze this crossing equation and discuss its consequences, restricting to real cross ratios lying within $0 < x < 1$.

\paragraph{Near the crossing symmetric point.}  Let us Taylor expand the right hand side of \eqref{LCBS} at the crossing symmetric point $x = 1/2$.  Since the right hand side is an odd function with respect to $x \to 1 - x$, all even power terms vanish.  The coefficients of the odd power terms to leading order at large $c$ give
\ie\label{BSEQ}
0 = \sum_{m \in {\cal S}_{1/2}} f'(m_{ext},m | 1/2 )^{2j-1} q_{\sigma_{ext}}(m)\widetilde g_0\left(m_{ext}, m | 1/2 \right) \quad \forall j \in {\mathbb N}.
\fe
Suppose the crossing equation is dominated by finitely many points, ${\cal S}_{1/2} = \{\widehat m_1,\widehat m_2,\cdots,\widehat m_n\}$, which we order by $\widehat m_1 < \widehat m_2 < \cdots < \widehat m_n$.  The fact that ${\cal S}_{1/2}$ is the set of global maxima means
\ie
S_{\sigma_{ext}}(\widehat m_1 | 1/2) = S_{\sigma_{ext}}(\widehat m_2 | 1/2) &= \dotsb = S_{\sigma_{ext}}(\widehat m_n | 1/2),
\fe
and this was used to factor out the exponential when going from \eqref{LCBS} to \eqref{BSEQ}.  By Property~\ref{NSfp} of the classical Virasoro block, $f'(m_{ext}, m | 1/2)$ is monotonically decreasing in $m$ and crosses zero exactly once, hence the equations \eqref{BSEQ} imply that the saddles must form pairs satisfying\footnote{First, $q(m) \widetilde g_0(m)$ does not vanish, otherwise $m$ would not appear in \eqref{BSEQ}.  Suppose $n > 1$.  In the large $j$ limit, by the monotonicity property of $f'(m)$, only $m_1$ and $m_n$ dominate the equation, and we conclude in
\eqref{equalfpq} for $k = 1$.  $m_1$ and $m_n$ drop out of \eqref{BSEQ}.  
Reiterate for other $k$.
}
%
\ie
\label{equalfpq}
f'(m_{ext}, \widehat m_{k} | 1/2) &=- f'(m_{ext}, \widehat m_{n+1-k} | 1/2),
\\
q_{\sigma_{ext}}(\widehat m_{k}) \widetilde g_0\left(m_{ext}, \widehat m_{k} | 1/2 \right) &=  q_{\sigma_{ext}}(\widehat m_{n+1-k}) \widetilde g_0\left(m_{ext}, \widehat m_{n+1-k} | 1/2 \right),
\fe
for $k = 1, \dotsc, [n/2]$.  Note that the last equation relates the one-loop (in $1/c$) part of the structure constants for pairs of saddles.  If $n$ is odd, then there is a lone saddle $\widehat m_{n+1\over2}$ sitting at the solution to $f'(m_{ext}, \widehat m_{n+1\over2} | 1/2) = 0$.

The multiplicity of the saddles is lifted in a small neighborhood ${1/2}-\epsilon < x < {1/2}+\epsilon$ of the crossing symmetric point.  The saddle with the largest $f'$ value dominates the region ${1/2}-\epsilon  < x < {1/2}$, and its partner which has the smallest $f'$ value dominates the region ${1/2} < x < {1/2}+\epsilon$.\footnote{Suppose $S_{\sigma_{ext}}(m|x)$ is a smooth function near $x=1/2$ and $m=\widehat m_k$ (the generalization to non-smooth $S_{\sigma_{ext}}(m|x)$ is simple).  It has an expansion at $x=1/2$,
\ie
S_{\sigma_{ext}}(m|x)&=S_{\sigma_{ext}}(\widehat m_k|1/2)+(x-1/2)\partial_xS_{\sigma_{ext}}(\widehat m_k|1/2)+{1\over 2}(m-\widehat m_k)^2\partial_m^2S_{\sigma_{ext}}(\widehat m_k|1/2)
\\
& \hspace{.5in} +(x-1/2)(m-\widehat m_k)\partial_m\partial_xS_{\sigma_{ext}}(\widehat m_k|1/2)+\cdots.
\fe 
When we move away from the crossing symmetric point, $x=1/2+\epsilon$, the new saddle point is at
\ie
m_\epsilon = \widehat m_k - {\partial_m \partial_x S_{\sigma_{ext}}(\widehat m_k|1/2)\over \partial_m^2 S_{\sigma_{ext}}(\widehat m_k|1/2)}\epsilon + {\cal O}(\epsilon^2),
\fe
and therefore
\ie
S_{\sigma_{ext}}(m_\epsilon|1/2+\epsilon) &= S_{\sigma_{ext}}(\widehat m_k|1/2)-{\epsilon\over 6}f'(m_{ext},\widehat m_k|1/2)+{\cal O}(\epsilon^2).
\fe
}

Focusing on a small neighborhood $1/2-\epsilon < x < 1/2+\epsilon$ but ignoring the possible multiplicity at the point $x = 1/2$, we conclude that there can be two scenarios (depending on whether $n = 1$ or $n \geq 2$ at $x = 1/2$).
\begin{enumerate}
\item  The four-point function is dominated by a single saddle at $m = 
\widehat m(m_{ext})$, solving the equation
\ie
\label{hatmm}
\boxed{
f'(m_{ext}, \widehat m(m_{ext}) | 1/2) = 0.}
\fe
In this case, the four-point function is smooth around $x=1/2$.  The solution $\widehat m(m_{ext})$ as a function of $m_{ext}$ is plotted in Figure~\ref{hatm}.
\item The four-point function is dominated by a saddle at $m = \widehat m_1$ for $1/2 - \epsilon < x < 1/2$ and another saddle at $m = \widehat m_2$ for $1/2 < x < 1/2 + \epsilon$, where $\widehat m_1$ and $\widehat m_2$ satisfy the relation
\ie
f'(m_{ext}, \widehat m_{1} | 1/2) &=- f'(m_{ext}, \widehat m_{2} | 1/2).
\fe
A phase transition occurs at $x = 1/2$.
\end{enumerate}
Next we prove the following proposition.
%
\begin{prop}
\label{Prop1}
{\it  The four-point function is dominated by saddles with weights $m \le \widehat m(m_{ext})$ for $x < 1/2$, and saddles with weights $m \ge \widehat m(m_{ext})$ for $x > 1/2$, where $\widehat m(m_{ext})$ is the unique solution to \eqref{hatmm}.  If there is a single saddle at $x = 1/2$, then its weight is $m = \widehat m(m_{ext})$.}
\end{prop}
\begin{proof}
Let us assume the contrary, that the four-point function at some cross ratio $x_* < 1/2$ is dominated by a saddle point with weight $m_* > \widehat m(m_{ext})$.  We recall the observed properties of the classical Virasoro blocks from earlier in this section.  Property~\ref{NSfp} implies that $\widehat m_1 \le \widehat m(m_{ext}) \le \widehat m_2$.  Property~\ref{NSfm2-fm1} implies that the four-point function in the entire range of cross ratios $x_* \le x< 1/2$ should be dominated by saddle points with weights $m \ge m_*$; in particular, this means that $\widehat m_1 \ge m_*$ in the neighborhood $1/2 - \epsilon < x < 1/2$.  Hence we arrive at contradicting inequalities.  
\end{proof}
%
The following lemma will be useful later.
\begin{cor}
\label{Cor1}  If the inequality
\ie
p_{\sigma_{ext}}(m) - {1 \over 6} f( m_{ext}, m | 1/2) \leq p_{\sigma_{ext}}(0) - {1 \over 6} f (m_{ext}, 0 | 1/2)
\fe
is obeyed for $m \le \widehat m(m_{ext})$, then it is obeyed for all $m\ge 0$.
\end{cor}
\begin{proof}
The contrary implies the existence of a classical branching ratio $S_{\sigma_{ext}}(m_*|x)$ at some weight $m_* > \widehat m(m_{ext})$ that is larger than $S_{\sigma_{ext}}(m|x)$ for all $m \leq \widehat m(m_{ext})$.  Then there is no saddle with weight $m \leq \widehat m(m_{ext})$, contradicting $\widehat m_1 \le \widehat m(m_{ext})$.
\end{proof}


%

%
\begin{figure}[t]
\centering
\includegraphics[width=.9\textwidth]{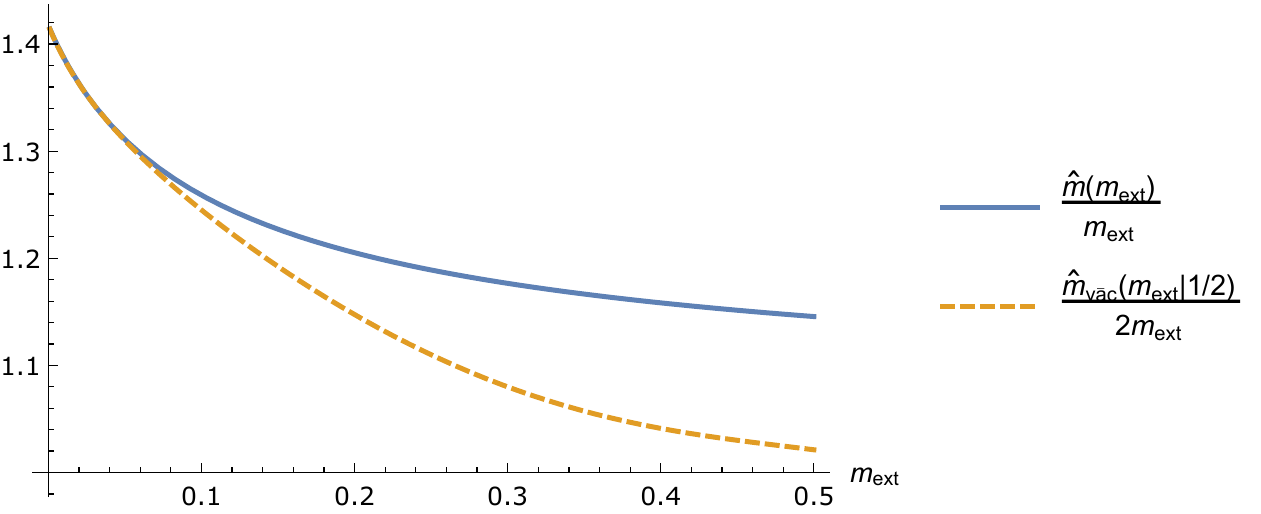}\hspace{.5in}
\caption{The ratios $\widehat m(m_{ext}) \over m_{ext}$ and $\widehat m_{\overline{vac}}(m_{ext}|1/2) \over 2m_{ext}$ as functions of the external weight $m_{ext}$.  See \eqref{hatmm} and \eqref{hatmvac} for definitions.}
\label{hatm}
\end{figure}

\paragraph{Away from the crossing symmetric point.}

At a generic cross ratio $x \neq 1/2$, the four-point function is dominated by a single saddle $m=\widehat m(x)$.  Here we ignore the measure zero set of cross ratios with multiple saddles.  Again Taylor expanding in $x$, we find that $\widehat m(x)$ and $\widehat m(1-x)$ must satisfy the relations\footnote{The $x$ in $\widehat m(x)$ and $\widehat m(1-x)$ are merely labels and should not be expanded.  More precisely, we first Taylor expand the crossing equation and then take the large $c$ limit.  The saddle condition is the same for all Taylor coefficients.
}
\ie
\label{equalSfpq}
f'(m_{ext},\widehat m(x)|x) &= -f'(m_{ext},\widehat m(1-x)|1-x),
\\
S_{\sigma_{ext}}(\widehat m(x) | x) &=  S_{\sigma_{ext}}(\widehat m(1-x) | 1-x),
\\
q_{\sigma_{ext}}(\widehat m(x))\widetilde g_0\left(m_{ext}, \widehat m(x) | x \right) &= q_{\sigma_{ext}}(\widehat m(1-x))\widetilde g_0\left(m_{ext}, \widehat m(1-x) | 1-x \right).
\fe

We point out a curious observation.  The $s$-channel block appearing in the crossing equation can be written via the fusion transformation \eqref{FusionIntegral} as an integral over $t$-channel blocks with different weights \cite{Ponsot:1999uf, Teschner:2001rv, Ponsot:2003ju}.  We show in Appendix~\ref{Saddle correspondence across channels} that for an $s$-channel block of weight $m \leq 1/24$ at a fixed cross ratio $x$, the fusion transformation is in fact dominated in semiclassical limit by the $t$-channel block whose weight is determined by equation \eqref{SaddleCorrespondenceEq}.  We find numerically that the solution to this equation coincides with the solution $\widehat m(1-x)$ to the first equation in \eqref{equalSfpq}.

\subsection{Universality of structure constants}
\label{sec:UniversalC}

A main result of the bootstrap is that both the classical $p_{\sigma_{ext}}(m)$ and one-loop $q_{\sigma_{ext}}(m)$ parts (in $1/c$) of the structure constants $C^2_{\sigma_{ext}\sigma_{ext}}(m)$ are related for the pair of dominant saddles $(\widehat m(x), \widehat m(1-x))$ at any cross ratio $0 < x < 1$, as is seen from the second and third equations in \eqref{equalSfpq}.

Let us consider a CFT whose {\bf spectrum of primaries has an order $c$ gap above the vacuum state},\footnote{More precisely, let us consider a sequence of CFTs labeled by $i = 1, 2, \dotsc$, with monotonically increasing central charges $c_i$, that admits a semiclassical limit.  For any given weight $h$, there exists an $I_h$ such that the only primary appearing in the OPE with weight below $h$ is the vacuum, for all $i \geq I_h$.  This is analogous to the condition in \cite{Hartman:2013mia} on the density of states.
}
so that $p_{\sigma_{ext}}(0) = 0$ and $q_{\sigma_{ext}}(0) = 1$.
The four-point function is dominated by the vacuum block near $x = 0$.  As the cross ratio is increased to some $x = x_{PT}$, this four-point function undergoes a phase transition and becomes dominated by a different saddle.  Let us denote by $\widehat m_{\overline{vac}}(m_{ext}, x)$, for $0 < x \leq 1/2$, the solution to
\ie\label{hatmvac}
\boxed{
f'(m_{ext},0|x) = - f'(m_{ext}, \widehat m_{\overline{vac}}(m_{ext}, x)|1-x),}
\fe
which is the $t$-channel saddle partner of the $s$-channel vacuum block.  Since $C^2_{\sigma_{ext}\sigma_{ext}}(0) = 1$ for the isolated vacuum block, $p_{\sigma_{ext}}(m)$ and $q_{\sigma_{ext}}(m)$ are unambiguously fixed for all $m > \widehat m_{\overline{vac}}(m_{ext}, x_{PT})$,
\ie
\label{UniversalC}
& \boxed{
p_{\sigma_{ext}}(\widehat m_{\overline{vac}}(m_{ext}, x)) =  {1 \over 6} f ( m_{ext}, \widehat m_{\overline{vac}}(m_{ext}, x) | 1 - x ) - {1 \over 6} f (m_{ext}, 0 | x ), }
\\
& \boxed{
 q_{\sigma_{ext}}(\widehat m_{\overline{vac}}(m_{ext}, x)) = {\widetilde g_0 (m_{ext},0 | x ) \over \widetilde g_0 (m_{ext}, \widehat m_{\overline{vac}}(m_{ext}, x) | 1-x)}.}
\fe
After the phase transition, even though the equations \eqref{equalSfpq} continue to relate pairs of saddles, we do not have an invariant reference point like the vacuum was before the phase transition, and therefore universality is lost.  If the only phase transition occurs at $x = x_{PT} = 1/2$, then this universality holds in the widest range $m \geq \widehat m_{\overline{vac}} (m_{ext}, 1/2)$.  The above analysis did not assume the positivity of the structure constants squared, but positivity is not violated by the universal formula \eqref{UniversalC} according to Property~\ref{Positiveg} of the one-loop Virasoro block.

Figure~\ref{hatm0} shows the function $\widehat m_{\overline{vac}}(m_{ext}, x)$ for $m_{ext}$ between $1/2400$ and $1/2$, and suggests that $\widehat m_{\overline{vac}}(m_{ext}, x) / m_{ext}$ is not very sensitive to $m_{ext}$.  Figure~\ref{universal} plots the universal classical and one-loop structure constants, $p_{\sigma_{ext}}(m)$ and $q_{\sigma_{ext}}(m)$.  High orders in the $x$-expansion are needed for the precision of results at large $m$, but the point here is universality.  Note that the structure constants $C^2_{\sigma_{ext}\sigma_{ext}}(m) \sim \exp(c\,p_{\sigma_{ext}}(m))$ decay faster than $16^{-mc}$, as is required by the convergence property of the four-point function \cite{Maldacena:2015iua}.

\begin{figure}[t]
\centering
\includegraphics[width=.7\textwidth]{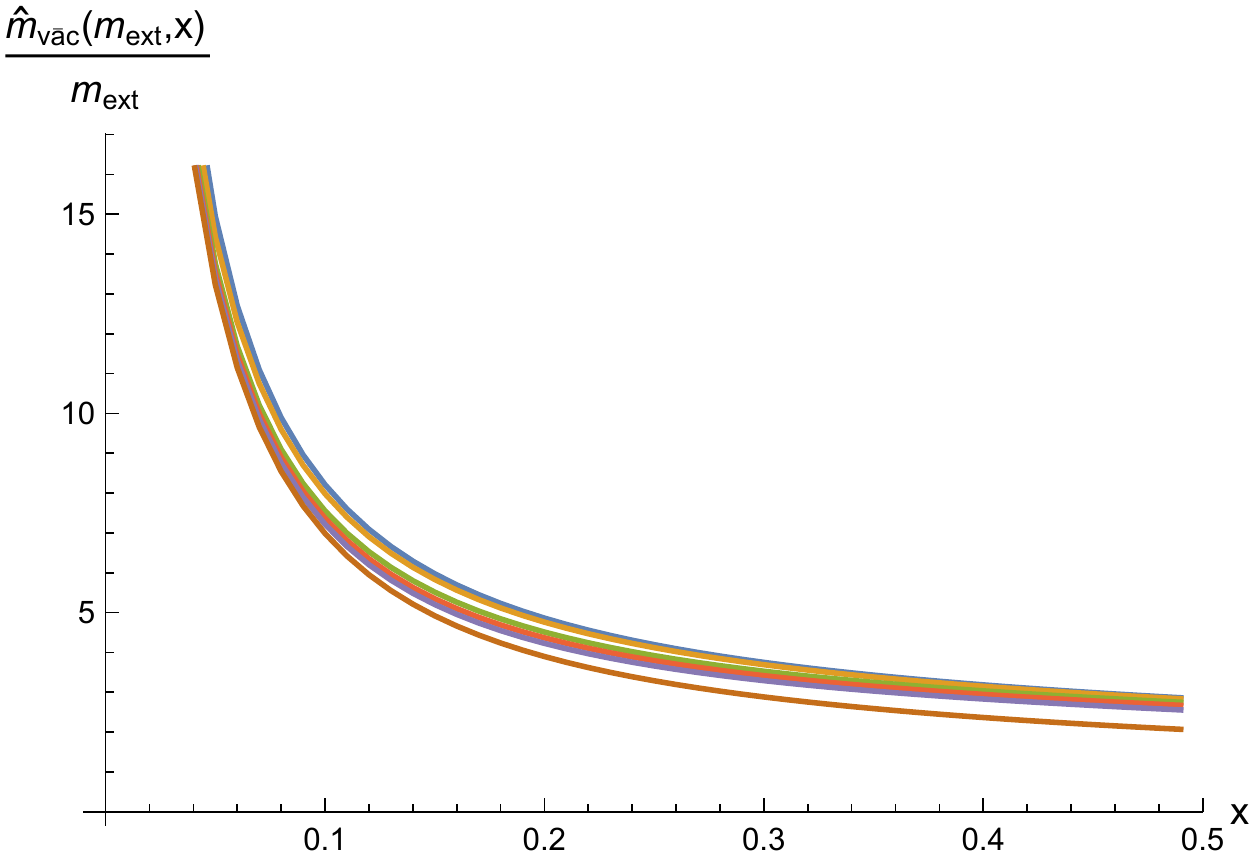}
\caption{The weight $\widehat m_{\overline{vac}}(m_{ext}, x)$
as a function of the cross ratio $x$ for external weights $m_{ext} = \A/24$.  See \eqref{hatmvac} for a definition.  The curves from top to bottom are for $\A = 1/100, 1/10, 1/2, 1, 2, 12$.}
\label{hatm0}
\end{figure}

If the external operators have a gapless OPE (the gap is of order $c^0$), then generically the $s$-channel saddle moves continuously away from the vacuum as $x$ is increased, until it reaches $\widehat m(m_{ext})$, which is the solution to Equation \eqref{hatmm}.  No sharp phase transition occurs ($x_{PT} = 0$).

Intuitively, the phase transition cross ratio $x_{PT}$ should be larger for theories with larger gaps.  However, even if the gap is large, as long as it is smaller than $\widehat m (m_{ext})$, we can tune the structure constants large to make $x_{PT}$ as small as we want.  For this reason, there does not seem to be a bound on $x_{PT}$ by the size of the gap.

Combining the above considerations with Lemma~\ref{Cor1}, we are led to the following propositions.
\begin{prop}
\label{Prop2a}
The gap (in the OPE of identical external operators) is bounded above by $m_{gap} \leq \widehat m_{\overline{vac}}(m_{ext}, 1/2)$.
\end{prop}

\begin{prop}
\label{Prop2b}
If the following condition is satisfied
\ie
\label{PropositionBound}
p_{\sigma_{ext}}(m) \leq {1 \over 6} f (m_{ext}, m | 1/2) - {1 \over 6} f( m_{ext}, 0 | 1/2) \quad \forall m\le\widehat m(m_{ext}),
\fe
then the only phase transition occurs at $x = 1/2$, and $p_{\sigma_{ext}}(m)$ and $q_{\sigma_{ext}}(m)$ follow the universal formula \eqref{UniversalC} for $m \geq \widehat m_{\overline{vac}}(m_{ext}, 1/2)$.
\end{prop}
The quantities $\widehat m(m_{ext})$ and $\widehat m_{\overline{vac}}(m_{ext}, 1/2)$ and are the unique solutions to the equations \eqref{hatmm} and \eqref{hatmvac}, and their numerical values are plotted in Figure~\ref{hatm}.  The entire discussion in this section can be easily generalized to include the anti-holomorphic sector.

\begin{figure}[H]
\centering
\subfigure[Classical $p_{\sigma_{ext}}(m)$]{
\includegraphics[width=.7\textwidth]{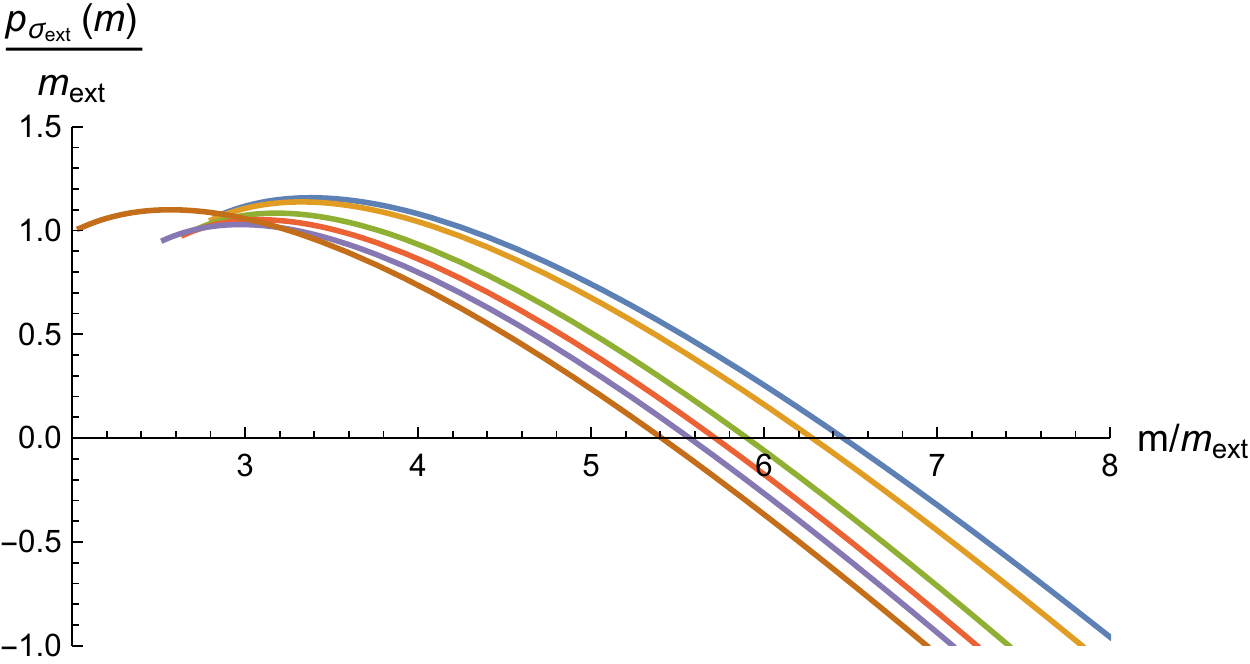}
}
\\
~
\\
~
\\
\subfigure[One-loop $q_{\sigma_{ext}}(m)$]{
\includegraphics[width=.7\textwidth]{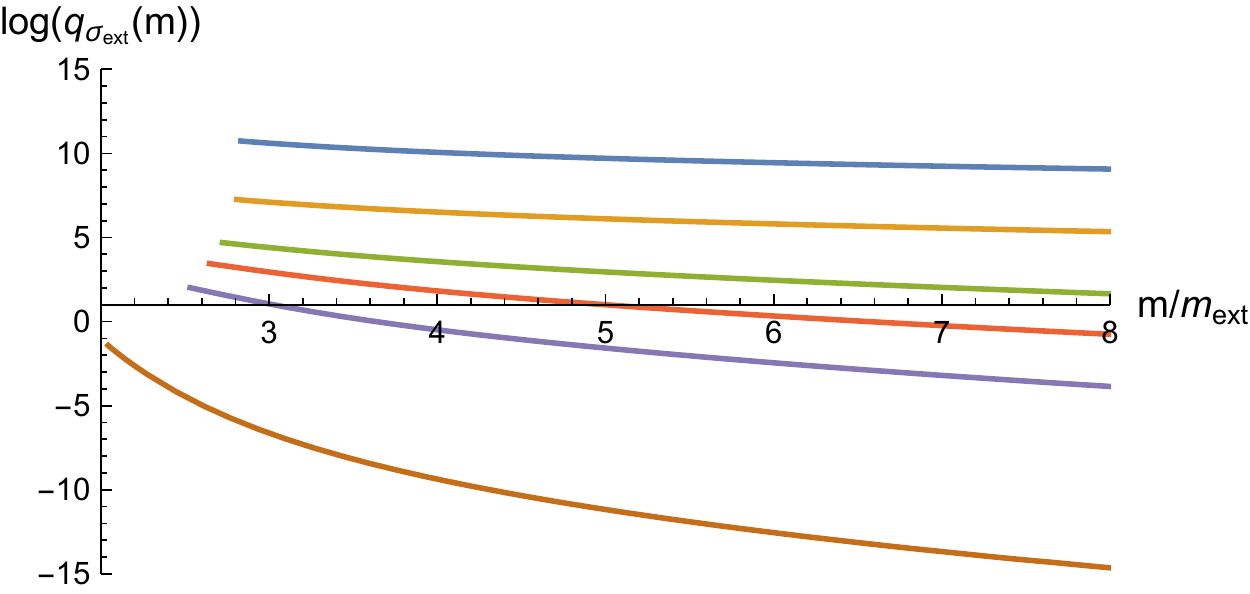}
}
\caption{The universal classical $p_{\sigma_{ext}}(m)$ and one-loop $q^{cont}_{\sigma_{ext}}(m)$ parts of the structure constants
as functions of the internal weight $m$, for external weights $m_{ext} = \A/24$.  See \eqref{UniversalC} for definitions.  The curves from top to bottom in both (a) and (b) are for $\A = 1/100, 1/10, 1/2, 1, 2, 12$.  
}
\label{universal}
\end{figure}

\section{Applications}
\label{Applications}
We examine a few theories in the semiclassical limit:  Liouville theory, product orbifold theories, and meromorphic CFTs.  Liouville theory and the untwisted sector four-point function in the product Ising model provide basic sanity checks of our results.  They both exhibit no phase transition, and at the crossing symmetric point, there is a single saddle whose weight is determined by Proposition~\ref{Prop1}.  We will explicitly see the movement of the dominant saddle as the cross ratio varies.

Twisted sector correlators in product orbifold CFTs are of various physical interests.  The semiclassical limit of product orbifold CFTs can be achieved in two ways, either by taking the number of copies to be large, or by taking the central charge of a single copy to be large.  The first limit is of interest in the symmetric product orbifold of ${\rm T}^4$ or K3, where the twisted sector states correspond to long strings in ${\rm AdS}_3 \times {\rm S}^3 \times ({\rm T}^4 {\rm~or~} {\rm K3})$ \cite{Dijkgraaf:1996xw}; a large number of copies gives a weakly coupled bulk description.  The second limit appears in the computation of higher genus partition functions and Renyi entropies in holographic theories \cite{Yin:2007gv,Headrick:2010zt, Hartman:2013mia}.  By considering the $\bZ_2$ twist field four-point function, we will recover the semiclassical version of the Hellerman bound on the gap in the spectrum of primaries \cite{Hellerman:2009bu}, and the logarithmically corrected Cardy formula that is valid for $h \geq c/12$ \cite{Cardy:1986ie,Hartman:2014oaa,Carlip:2000nv}.  Furthermore, we give a condition for there to be a single phase transition in the second Renyi entropy, which was argued to be true in holographic theories by \cite{Headrick:2010zt, Hartman:2013mia}.

\subsection{Liouville theory}

\begin{figure}[t]
\centering
\includegraphics[width=.7\textwidth]{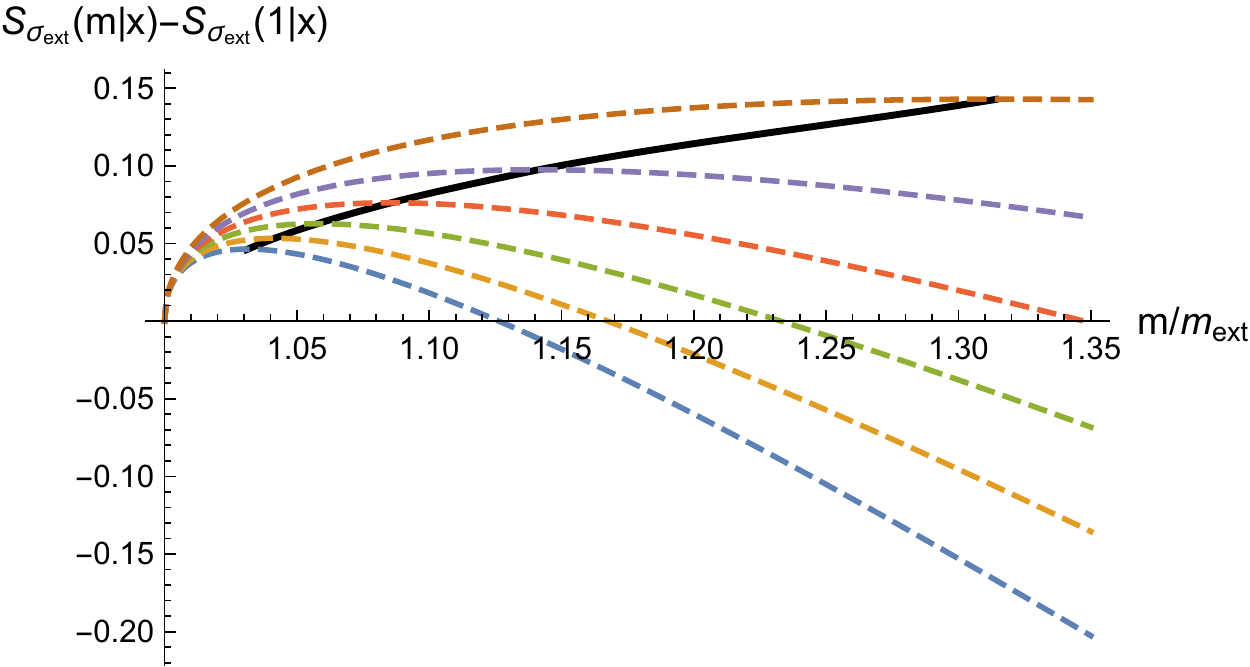}
\caption{The ground state $m_{ext} = 1/24$ four-point function in Liouville theory.  The dashed lines plot the classical branching ratio $S_{\sigma_{ext}}(m|x)$ (defined in \eqref{BranchingRatio}) as a function of the internal weight $m$, for cross ratios $x = 10^{(\A-5)/10}/2$ with $\A = 0, 1, \dotsc, 5$ from bottom to top.  The solid line traces the dominant saddle as the cross ratio is varied.  The dominant saddle is at $m = 1.32 \, m_{ext}$ (semiclassical: $\widehat m(m_{ext}) = 1.32 \, m_{ext}$) at the crossing symmetric point.}
\label{figliouville}
\end{figure}

Liouville theory is the simplest example of a CFT with a semiclassical limit.\footnote{See \cite{Nakayama:2004vk} for a review of Liouville theory.}  It does not contain a vacuum state, and the spectrum of primaries is continuous above the ground state of weight $h_{ground} = (c-1)/24$.  A closed form formula for the structure constants was proposed in \cite{Dorn:1992at,Dorn:1994xn,Zamolodchikov:1995aa,Teschner:1995yf}, and was mathematically proven to satisfy crossing symmetry in \cite{Ponsot:1999uf,Teschner:2001rv,Teschner:2001hk}.  In fact, many properties of the Virasoro blocks were discovered in the study of Liouville theory \cite{Ponsot:1999uf,Teschner:2001rv,Teschner:2001hk}.  Here we use Liouville theory to check the results of our semiclassical bootstrap analysis.

Consider the four-point function of identical operators of weight $h_{ext} = m_{ext} c$ in the semiclassical limit.  Since a vacuum is absent, we expect that the dominant saddle should move continuously from $m_{ground} = 1/24$ to $\widehat m(m_{ext})$ (the unique solution to \eqref{hatmm}) as we vary the cross ratio from $x=0$ to $1/2$.

In the semiclassical limit, the continuous spectrum of primaries in Liouville theory is parameterized by
\ie
\eta &= {1\over 2} - \sqrt{{1\over 4}-6m} \in {1 \over 2} + i \bR_{\geq0}, \quad m \geq {1 \over 24}.
\fe
The structure constants reduce to the on-shell classical Liouville action on a three-punctured sphere \cite{Zamolodchikov:1995aa}
\ie
C^2_{\eta_{ext}, \eta_{ext}}(\eta) = \exp \left[ {2\over b^2} \,{\rm Re}\,S^{(cl)}(\eta_{ext},\eta_{ext},\eta) \right],
\fe
where
\ie
\text{Re}\,S^{(cl)}(\eta_{ext}, \eta_{ext}, \eta) &= -H(2\eta_{ext} + \eta - 1) -H(2\eta_{ext}-\eta ) - 2H(\eta) +H(0)
\\
& \hspace{.5in} + H(2\eta_{ext}) + H(2\eta_{ext}-1) +{ H(2\eta) + H(2\eta-1) \over 2},
\fe
and $H(\eta) = G(\eta) + G(1-\eta) = \int^\eta_{1\over 2}\log\gamma(x)dx$ is the semiclassical limit of the special function $b^2 \Upsilon_b$ (see Appendix~\ref{Semiclassical limit of the special functions}).

At a fixed cross ratio $x$, the four-point function is dominated by a single saddle that solves\footnote{Note that this equation is exactly the same equation that determines the dominant $t$-channel saddle \eqref{SaddleCorrespondenceComplex} in the fusion transformation.  The saddle point analysis of classical Liouville theory was previously considered in \cite{Hadasz:aa}.
}
\ie
{\partial \over \partial m} \Big[ {\rm Re}\,S^{(cl)}(\eta_{ext}, \eta_{ext}, \eta) - f(m_{ext}, m | x) \Big] = 0.
\fe
Figure~\ref{figliouville} shows the distribution of the classical branching ratio (defined in \eqref{BranchingRatio}) as the cross ratio varies.  The solution at $x = 1/2$ is numerically verified to be equal to $m = \widehat m(m_{ext})$, as is required by conformal symmetry.

\subsection{Product Ising model}

\begin{figure}[t]
\centering
\includegraphics[width=.7\textwidth]{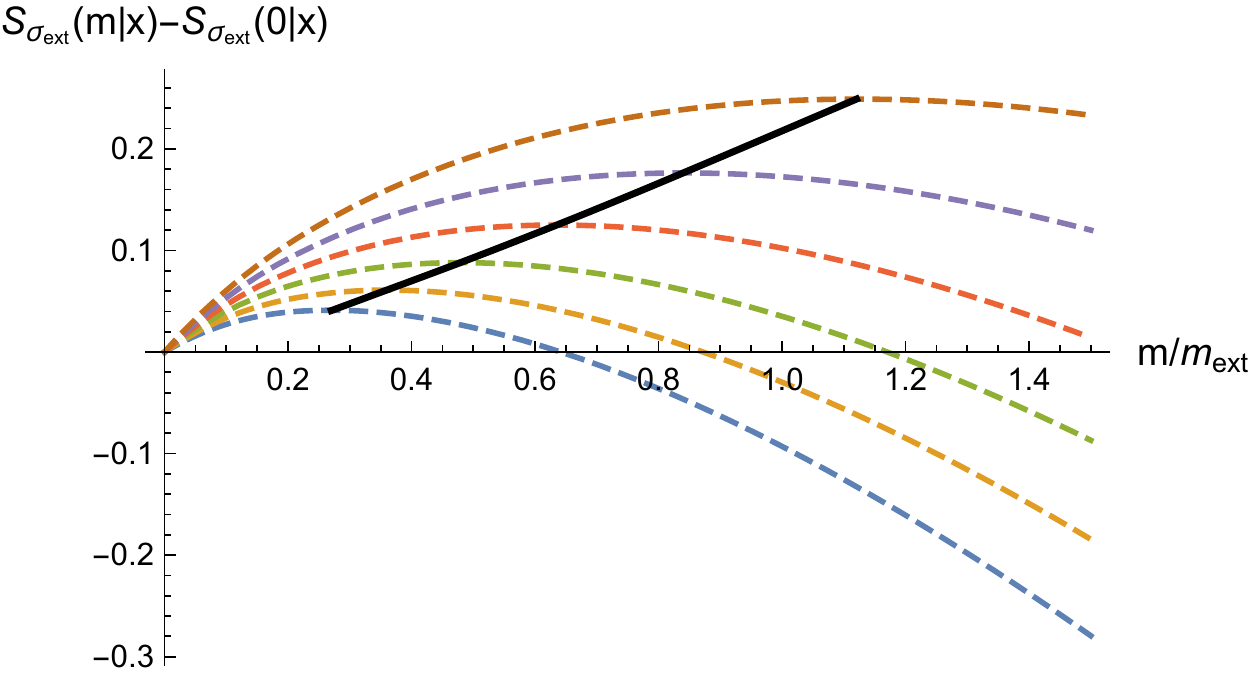}
\caption{The $\sigma^{64}$ ($m_{ext} = 1/8$) four-point function in the product Ising model.  The dashed lines plot the classical branching ratio  $S_{\sigma_{ext}}(m|x)$ (defined in \eqref{branch}) for scalars as a function of the internal weight $m$, for cross ratios $x = 10^{(\A-5)/10}/2$ with $\A = 0, 1, \dotsc, 5$ from bottom to top.  The solid line traces the dominant saddle as the cross ratio is varied.  At the crossing symmetric point, the dominant weight is at $m = 1.12 \, m_{ext}$.  It further approaches the semiclassical value $\widehat m(m_{ext}) = 1.24 \,m_{ext}$ as the number of copies is increased.}
\label{figising}
\end{figure}

Consider the product of $n$ copies of the Ising model, which has central charge $c = n/2$.  The four-point function of the product spin field $\sigma_{ext} = \sigma^n$, which has weight $m_{ext} = \bar m_{ext} = {1 / 8}$, is the $n$-th power of the single copy four-point function
\ie
F(x, \bar x) = |x(1-x)|^{-1/4} \left( \left| {1+\sqrt{1-x} \over 2} \right| + \left| {1-\sqrt{1-x} \over 2x} \right| \right).
\fe
The structure constants can be obtained by decomposing $F(x, \bar x)^n$ into Virasoro blocks (of finite central charge).  At large $n$, we expect the behavior of the structure constants to obey our results from the semiclassical bootstrap.  Figure~\ref{figising} shows the distribution of the classical branching ratio for scalars, defined as in \eqref{BranchingRatio} but with a smoothly interpolated\footnote{Here the classical Virasoro block $f$ includes both the holomorphic and anti-holomorphic factor.}
\ie
\label{branch}
S_{\sigma_{ext}}(m|x) = p_{\sigma_{ext}}(m) - {f(m_{ext}, m|x) \over 6} = c^{-1} \log C^2_{\sigma_{ext}\sigma_{ext}}(m) - {f(m_{ext}, m|x) \over 6},
\fe
as the cross ratio is varied.  The dominant saddle moves continuously from the ground state $m = 0$ to $m = \widehat m(m_{ext})$ as the cross ratio $x$ varies from 0 to 1/2.

\subsection{$\bZ_2$ orbifold, character expansion, and Renyi entropy}
\label{Z2Twist}

The four-point function of $\bZ_2$ twist fields in the symmetric product orbifold can be lifted to a torus partition function, the modular invariance of which can be written in the form of a crossing equation.  Denote by $q(x) = \exp(i \pi \tau(x))$ the elliptic nome of $x$, and by $\cal P$ the function\footnote{The elliptic nome  is defined by
\ie
\log q(x) \equiv - {\pi K(1-x) \over K(x)}.
\fe
}
\ie
{\cal P}(h_{ext}, h, c | x) = { (16 q)^{h-{(c-1) / 24}} (x(1-x))^{{(c-1) / 24} - 2h_{ext}} K(x)^{{(c-1) / 4} - 8 h_{ext}} },
\fe
where $K(x) \equiv {}_2F_1(1/2, 1/2, 1 | x)$ is a hypergeometric function.  The function ${\cal P}$ is the prefactor in the elliptic representation of the Virasoro block\footnote{This form appears in the Zamolodchikov recurrence relation \cite{Zamolodchikov:1985ie,Zamolodchikov:1987}.
}
\ie
\label{elliptic}
{\cal F}(h_{ext}, h, c | x) = {\cal P}(h_{ext}, h, c | x) H(h_{ext}, h, c | q),
\fe
where $H(h_{ext}, h, c | q) = 1 + {\cal O}(q)$.  The non-vacuum character $\chi$ is related to $\cal P$ by the identity
\ie
\label{mathidentity}
\chi(q) = {q^{2(h - {(c - 1) / 24})} \over \eta(\tau) } = 16^{-2(h - {c / 24})}  (x (1 - x))^{c / 24} \, {\cal P}({c / 16}, 2h, 2c-1 | x),
\fe
$\eta(\tau)$ being the Dedekind eta function.  The vacuum character is $\chi^{vac}(q) = (1-q^2) \chi(q)$.  The modular transform $\tau \to - {1 / \tau}$ then translates to crossing $x \to 1-x$ under this identity.

A physical meaning of this equivalence was explained in \cite{Headrick:2010zt}.  Given any CFT $\cal C$ with central charge $c$, we can take the symmetric product orbifold ${\rm Sym}^2 \cal C$ and consider the four-point function of the twist field ${\cal E}$ which has weight $h = {c/16}$ \cite{Lunin:2000yv}.  This four-point function has a lift to the torus partition function $Z(q)$ of $\cal C$, and also computes the second Renyi entropy of two intervals \cite{Lunin:2000yv,Headrick:2010zt}.  Expanding the torus partition function in characters is equivalent to expanding the twist field four-point function in ``${\rm Sym}^2(Vir)$ blocks'' of primary operators of the form $\sigma_{\cal C} \otimes \sigma_{\cal C}$, where $\sigma_{\cal C}$ are primaries in ${\cal C}$.  Note that the ${\rm Sym}^2(Vir)$ descendants of such an operator include infinitely many Virasoro primaries.  It was checked in \cite{Headrick:2010zt} to the first few orders in the $x$-expansion that the ${\rm Sym}^2(Vir)$ blocks are indeed equal to the characters up to a conformal factor.

We presently explain how to apply our results from the semiclassical bootstrap.  Observing that in the semiclassical limit, the function $H$ multiplying the prefactor ${\cal P}$ in the Virasoro block \eqref{elliptic} does not contribute to the classical Virasoro block\footnote{This would not be true if the external weight did not scale as $c/16$.
}
\ie
\lim_{c \to \infty} {\log H({c / 16}, 2 m c, 2c-1 | q) \over c} \to 0,
\fe
we obtain the following identity ($m = h/c$)
\ie
\label{CharacterIdentity}
\left( m - {1 \over 24} \right) \log q^2 &= \left[ {\log 16 \over 12} + {\log(x(1-x)) \over 24} \right]
\\
& \hspace{.5in} + \left[ - {f(m_{ext} = {1/32}, m | x) \over 3} - 2 m \log 16 \right].
\fe
On the left is the classical character, and on the right, the first bracket is a conformal factor, and the second bracket is the classical ${\rm Sym}^2 (Vir)$ block.  We see that in this normalization of the ${\rm Sym}^2 (Vir)$ block, each $\sigma_{\cal C} \otimes \sigma_{\cal C}$ appears in the twist field four-point function with unit coefficient.  Therefore, when decomposing the twist field four-point function with respect to the ${\rm Sym}^2 (Vir)$ blocks, the expansion coefficients $C^2_{{\cal E}{\cal E}}(m)$ are precisely the classical density of primaries in the single copy CFT $\cal C$,
\ie
c\,\rho^P(h) = C^2_{{\cal E}{\cal E}}(m) = \exp[c\,p_{\cal E}(m)] q_{\cal E}(m).
\fe
The factor of $c$ on the left comes from the difference between the measures $dh$ and $dm$.

We can pretend that we are bootstrapping with the classical Virasoro block $f(m_{ext} = {1/32}, m | x)$ by defining an effective classical structure constant 
\ie
p'_{\cal E}(m) \equiv { p_{\cal E}(m) \over 2} - m \log 16.
\fe
Then by Proposition~\ref{Prop2a}, a bound on the gap in the spectrum of primaries in ${\cal C}$ is given by\footnote{Using the identities
\ie
{ q'(x) \over q(x) } &= - {\pi^2 \over 4 x (x-1) K^2(x)}, \quad \log q(x) \log q(1-x) = \pi^2,
\fe
we can show that the $t$-channel saddle parter of the $s$-channel vacuum is at
\ie
\widehat m_{\overline{vac}}(1/32, x) &= {1 \over 24} \left( 1 + {K^2(1-x) \over K^2(x)} \right) = {1 \over 24} \left[ 1 + \left( {\log(q(x)) \over \pi} \right)^2 \right].
\fe
The other formulae in this section are easily derived with the use of these identities.
}
\ie
m_{gap} \leq \widehat m_{\overline vac}(m_{ext} = 1/32, 1/2) = {1 \over 12}, \quad {\rm or} \quad h_{gap} \leq {c \over 12} + {\cal O}(c^0),
\fe
which is the holomorphic version of the Hellerman bound \cite{Hellerman:2009bu}.\footnote{After taking into account the anti-holomorphic sector, we obtain the conventional Hellerman bound \cite{Hellerman:2009bu} for the total weight, $\Delta_{gap} \leq {c / 6}$.
}
Furthermore, if the condition
\ie
\label{Z2pprime}
p'_{\cal E}(m) \leq m (\pi - \log 16) \quad \forall m \le \widehat m(m_{ext} = 1/32) = {1 \over 24}
\fe
is satisfied, then there is only one phase transition at the crossing symmetric point, and $p'_{\cal E}(m)$ obeys (by Proposition~\ref{Prop2b} and Lemma~\ref{Cor1})
\ie
p'_{\cal E}(m) &= {\pi \over \sqrt{6}} \sqrt{m - {1 \over 24}} - m \log 16 \quad \forall m \geq {1\over12},
\\
p'_{\cal E}(m) &\leq m (\pi - \log 16) \quad \forall m \geq 0.
\fe 
Next, the ${\cal O}(c^0)$ part of the non-vacuum character identity \eqref{mathidentity} reads
\ie
{q^{1/12} \over \eta(\tau)} &= { (16 q)^{1/12} \over (x(1-x))^{1/12} K(x)^{1/2} }.
\fe
We can pretend that this is the generic (non-vacuum) one-loop block $g(m_{ext} = {1/32}, m > 0 | x)$.  Assuming that spectrum of primaries has an order $c$ gap above a the vacuum state, the vacuum one-loop block is $g(m_{ext} = {1/32}, m = 0 | x) = (1-q^2) g(m_{ext} = {1/32}, m > 0 | x)$.  Then if \eqref{Z2pprime} holds, the one-loop structure constants obey the universal formula
\ie
q_{\cal E}(m) &= \sqrt{12c\over 24m - 1} \left(1 - e^{-2\pi(24m-1)} \right) \exp\left[ -{\pi \over 6}  {12m - 1 \over \sqrt{24m - 1}} \right]  \quad \forall m \geq {1 \over 24}.
\fe

Translating the above into a statement about the density of primaries, we obtain the next proposition.
\begin{prop}
\label{Prop3}
If the spectrum of primaries has an order $c$ gap above a the vacuum state, and the light spectrum is sparse in the sense of
\ie
\rho^P(h) \leq \exp(2\pi h) \quad \forall h \le {c \over 24},
\fe
then this inequality holds for all $h \geq 0$.  Furthermore, 
the density of primaries for the heavy spectrum $h \geq {c / 12}$ is given by\footnote{A Cardy formula analogous to \eqref{supercardy} but for the density of {\it all} states $\rho(h)$ can be obtained by the convolution
\ie
\rho(h)=\sum_n\rho^P(h-n)p(n)-\sum_n p(n)\delta(1+n-h).
\fe
The result is
\ie
\rho(h)={1\over\sqrt{c}}\exp\left[{2\pi} \sqrt{{c \over 6} \left( h - {c \over 24} \right)}\right]\left[\left(c^3\over 96(h-c/24)^3\right)^{1\over 4}\prod_{k=2}^\infty{1\over \left(1-e^{-2\pi k\sqrt{24h/c-1}}\right)}+{\cal O}(1/c)\right]
\fe
for $h \geq c /12$, in the semiclassical $1/c$ expansion with $h/c$ fixed.
}
\ie
\label{supercardy}
\rho^P(h) &= \sqrt{12 \over 24h-c} ~ \left( 1 - e^{-2\pi\sqrt{24h/c-1}} \right)
\\
& \hspace{.5in} \times \exp\left[ {2\pi} \sqrt{{c \over 6} \left( h - {c \over 24} \right)} - {\pi\over 6} \left( {12h - c \over \sqrt{c\,(24h - c)}} \right) + {\cal O}(1/c) \right]
\fe
in the semiclassical $1/c$ expansion with $h/c$ fixed.
\end{prop}

This takes the form of a logarithmically corrected Cardy formula for the density of primaries.  The Cardy formula {\it for primaries} is related to the original formula for the full spectrum by a shift of $c \to c - 1$.  Logarithmic corrections are obtained by a slight modification of \cite{Carlip:2000nv} to be
\ie
\rho^P_{Cardy}(h) &= \sqrt{12 \over 24 h - (c-1)} \exp\left[ 2\pi \sqrt{{c-1 \over 6} \left( h - {c-1\over24} \right)} \right], \quad h \gg c.
\fe
The semiclassical expansion of this formula almost agrees with~\eqref{supercardy} in Proposition~\ref{Prop3}, except for the factor of $1 - e^{-2\pi\sqrt{24h/c-1}}$ that is exponentially suppressed in the Cardy regime $h \gg c$.  That the Cardy formula is also valid for $h \geq c/12$ at large central charges was first discovered in \cite{Hartman:2014oaa}.

\paragraph{Comments on Renyi entropies.}  The four-point function of the $\bZ_2$ twist field computes the second Renyi entropy of two intervals, whereas the four-point functions for the maximal twist fields in the $\bZ_n$ product orbifolds compute higher Renyi entropies.  The results of \cite{Headrick:2010zt, Hartman:2013mia} suggest that in CFTs with a weakly coupled holographic dual, all Renyi entropies should have a single phase transition at the crossing symmetric point.  They argued that this is true in the $\bZ_2$ case assuming a sparse light spectrum, but for higher Renyi entropies it was left as still an open question.  Proposition~\ref{Prop3} makes precise the condition of a sparse spectrum, while Proposition~\ref{Prop2b} gives a condition for there to be a single phase transition.

\subsection{Meromorphic CFTs}
\label{Mero}

Consider the four-point function of holomorphic conserved currents $\sigma_{ext}$ of integer weight $h_{ext}$.  Meromorphy fixes the functional form to be \cite{Headrick:2015gba}
\ie
\label{MeromorphicForm}
F(x) = { \sum_{i = 0}^{4h_{ext}} a_i x^i \over x^{2h_{ext}} (1-x)^{2h_{ext}} },
\fe
which depends on $4 h_{ext} + 1$ coefficients $a_i$.  After imposing crossing symmetry, we are left with $\left[ {4h_{ext} + 4 \over 6} \right] + (\delta_{h_{ext} \!\!\! \mod 6})$ many coefficients.  The division by 6 in the first term can be understood as the order of the crossing group $S_3$, while the second term is due to accidental symmetries when $h_{ext} \in 6\bZ$.

In a CFT $\cal C$ with charge $c$, we can decompose $F(x)$ into Virasoro blocks.  Because the four-point function $F(x)$ only receives contributions from primaries of even integer weights, just from counting the freedom of tuning the coefficients, the gap in the primary spectrum of $\sigma \times \sigma$ is bounded above by \cite{Bouwknegt:1988sv,Keller:2013qqa}
\ie
\label{MeromorphicGap}
h_{gap} \leq 2 \left( \left[ {4h_{ext} + 4 \over 6} \right] + (\delta_{h_{ext} \!\!\!\! \mod 6}) + 1 \right),
\fe
which asymptotes to $4h_{ext} / 3$ at large $h_{ext}$.

\paragraph{$\bZ_2$ twist field four-point function.}  Suppose $\cal C$ is the $\bZ_2$ symmetric product orbifold of a meromorphic CFT $\cal B$ of central charge $c$, and let $\sigma_{ext} = {\cal E}$ be the twist field of weight $h_{ext} = (2c)/32$.  The naive upper bound \eqref{MeromorphicGap} on the gap in ${\cal E} \times {\cal E}$, coming from counting the number of tunable coefficients in \eqref{MeromorphicForm}, is $gap_{{\cal E} \times {\cal E}} \leq (2c)/24$.  By lifting to the torus, this translates in the semiclassical limit to an upper bound on the gap in the spectrum of primaries in $\cal B$, that is the extremal bound for meromorphic CFTs:  $h^{\cal B}_{gap} \leq c/24 + {\cal O}(c^0)$ \cite{Witten:2007kt}.

\begin{figure}[t]
\centering
\includegraphics[width=.7\textwidth]{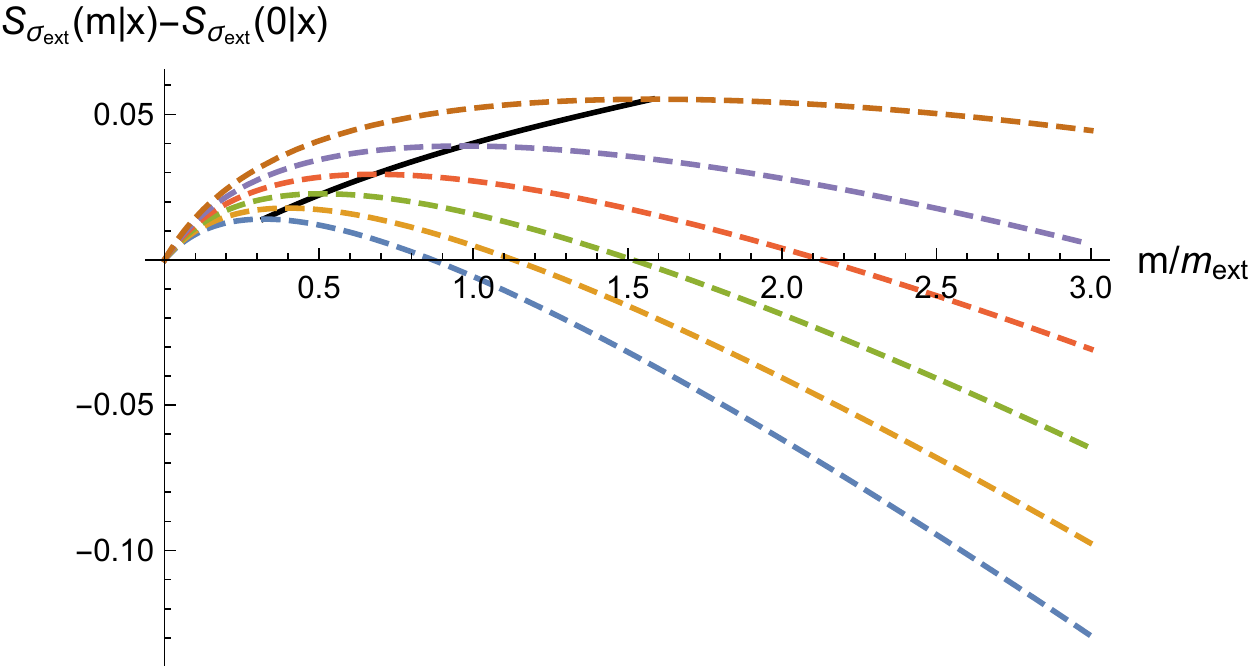}
\caption{The $h_{ext} = 2$ four-point function in the  monster theory.  The dashed lines plot the smoothly interpolated classical branching ratio $S_{\sigma_{ext}}(m|x)$ (defined in \eqref{branch})  as a function of the internal weight $m$, for cross ratios $x = 10^{(\A-5)/10}/2$ with $\A = 0, 1, \dotsc, 5$ from bottom to top.  The solid line traces the dominant saddle as the cross ratio is varied.  The dominant saddle is at $m = 1.58 \, m_{ext}$ (semiclassical: $\widehat m(m_{ext} = 1/24) = 1.32 \, m_{ext}$ or $\widehat m(m_{ext} = 1/12) = 1.27 \, m_{ext}$) at the crossing symmetric point.}
\label{figmonster}
\end{figure}

\paragraph{Extremal CFTs.}  Extremal meromorphic CFTs have central charge $c = 24k$ and a gap of size $h^{(k)}_{gap} = k + 1$, for $k \in \mathbb{N}$.  We take a sequence of operators ${\cal O}^{(k)}$ with weight $h^{(k)}_{ext} = h^{(k)}_{gap}$, and consider the four-point function $\langle {\cal O}^{(k)} {\cal O}^{(k)} {\cal O}^{(k)} {\cal O}^{(k)} \rangle$.  Assuming that this four-point function has a semiclassical limit, as the cross ratio $x$ varies, the dominant saddle cannot move continuously away from the vacuum due to the large gap, and there should be a phase transition at a finite $x = x_{PT}$.  Then for sufficiently large weight $m > \widehat m_{\overline{vac}}(m_{ext}, x_{PT})$, the structure constants should follow the universal formula \eqref{UniversalC}.  Figure~\ref{figmonster} shows the smoothly interpolated classical branching ratio (defined in \eqref{branch}) for the $k = 1$ monster theory, whose four-point function is known explicitly \cite{frenkel1985moonshine}.  Since the gap is at  $m_{gap} = m_{ext}$, the phase transition occurs between the bottom two curves at $x_{PT} \approx 0.16$.  We do not know whether this $c = 24$ picture is actually representative of the semiclassical limit.

\section{Comments on gravity}
\label{gravity}

This section discusses aspects of classical Virasoro blocks in the context of holography.  We first review the worldline prescription that reproduces classical Virasoro blocks in the ``heavy-light'' limit.  We then propose a similar prescription in the ``light'' limit, and discuss the implications of the results of semiclassical bootstrap.

\paragraph{Bulk dual of classical Virasoro blocks in the ``heavy-light'' limit.}

In \cite{Fitzpatrick:2014vua,Hijano:2015rla,Hijano:2015qja}, it was shown that Virasoro blocks in the semiclassical ``heavy-light'' limit are dual to certain worldline actions in a conical defect or BTZ black hole background.  More precisely, in the regime where the weights ($h_i = m_i c$) all scale with the central charge $c$, and $m_3 = m_4 = m_h$ are of order one, but $m_{1}$, $m_{2}$ and $m$ are parametrically small, we can treat the ``light'' operators $\sigma_1$ and $\sigma_2$ as probes of the background created by the ``heavy" operators $\sigma_3$ and $\sigma_4$.  The heavy operators create a bulk geometry that is either a conical defect ($m_h < 1/24$) or a BTZ black hole ($m_h\ge1/24$), and the leading order expansion in $m_{1}, m_2$ of the classical Virasoro block $f(m_1, m_2, m_h, m_h, m|x)$ can be computed by minimizing a worldline action.  The worldline action consists of the geodesic distance from $\sigma_{1}$ on the boundary to a bulk point ${\bf x}$, weighted by $m_{1}$, and the same for $\sigma_2$, plus the geodesic distance from the bulk point ${\bf x}$ to the conical singularity or the BTZ black hole horizon, weighted by $m$.  The position of the bulk point ${\bf x}$ is chosen to minimize this worldline action.

\paragraph{Bulk dual of classical Virasoro blocks in the ``light'' limit.}
Still in the semiclassical limit, consider a different parameter regime, where all weights $m_i$ and $m$ are parametrically much smaller than one.  We expect a similar correspondence between the leading order expansion of the classical Virasoro block $f(m_1,m_2,m_3,m_4,m|x)$ in $m_i, m$, and a worldline action in the AdS$_3$ background.  It is simplest to work in a Poincare patch of AdS$_3$ with metric $ds^2=(dy^2+dxd\bar x)/ y^2$.  The geodesic distance $L({\bf x},{\bf x}')$ between two bulk points $\bf x$ and ${\bf x}'$ is given by
\ie
&L({\bf x},{\bf x}')=\cosh^{-1}(1+u({\bf x},{\bf x}')), \quad u({\bf x},{\bf x}')={(y-y')^2+|x-x'|^2\over 2yy'},
\fe
which diverges as one take the bulk point ${\bf x}'$ to the boundary.  After regularizing this divergence,\footnote{The geodesic distance expanded in $1/y'$ is given by
\ie
L({\bf x},{\bf x}')=\log\left[{y^2+|x-x'|^2\over y}\right]-\log(y')+{\cal O}(y').
\fe
We simply drop the divergent logarithm.}
the geodesic distance from a bulk point ${\bf x}$ to a boundary point $x'$ is
\ie
L({\bf x},x')=\log\left[{y^2+|x-x'|^2\over y}\right].
\fe

For simplicity, we choose identical masses $m_i = m_{ext}$, and consider the worldline action
\ie
& S(x_1,x_2,x_3,x_4)
\\
&= \min_{{\bf x}_a,{\bf x}_b}\left\{ m_{ext} \left[ L({\bf x}_a,x_1)+L({\bf x}_a,x_2)+L({\bf x}_b,x_3)+L({\bf x}_b,x_4) \right] + mL({\bf x}_a,{\bf x}_b) \right\}.
\fe
\begin{figure}[t]
\centering
\includegraphics[width=.4\textwidth]{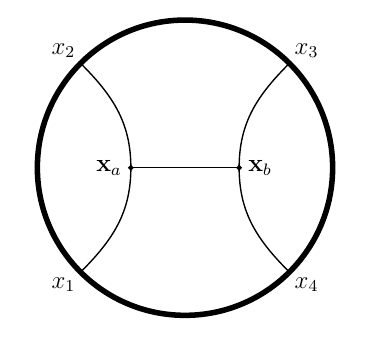}
\caption{The worldline action corresponding to the classical Virasoro block.  The bulk points $\bf x_a$ and $\bf x_b$ are chosen to minimize the total geodesic distance.}
\label{geodesic}
\end{figure}
\hspace{-0.09in}We propose the following relation between the worldline action $S(x_1,x_2,x_3,x_4)$ and the classical Virasoro block $f(m_{ext},m|x)$
\ie\label{ftoS}
&{ {\rm Re}\,f(m_{ext},m|x) \over 6} + \text{($x$-independent term)}
\\
&= - m_{ext} \left[ \log|x_1-x_4|^2 + \log|x_2-x_3|^2 \right]+S(x_1,x_2,x_3,x_4) + {\cal O}(m_{ext}, m)^2.
\fe
By conformal symmetry, we can fix the four points on the boundary at $x_1=-z/2$, $x_2=z/2$, $x_3=-1$, $x_4=1$.  Then by the symmetry of the system, the two bulk points $\bf x_a$, $\bf x_b$ that minimize the worldline action $S(x_1,x_2,x_3,x_4)$ must be located at $x_a=x_b=0$.  Further minimizing with respect to the two remaining variables $y_a$ and $y_b$, we find that the following solution exists as long as the triangular inequality $2m_{ext}>m$ is obeyed,
\ie\label{Szz11}
S(-z/2,z/2,-1,1) &= m\,\cosh^{-1}\left[ 4(2m_{ext}-m)^2+(2m_{ext}+m)^2 z^2\over 4(4m_{ext}^2-m^2)z\right]
\\
&\hspace{1.5in} +m_{ext}\log[64 z^2]-2m_{ext}\log[4m_{ext}^2-m^2].
\fe
Expanding in the cross ratio $x=8z/(2+z)^2$, we find that \eqref{ftoS} is indeed satisfied.

\paragraph{Relation to Ryu-Takayanagi formula.}
The $m=0$ version of the classical block/worldline action correspondence was used in \cite{Hartman:2013mia} to match the entanglement entropy of two intervals in the boundary CFT with the Ryu-Takayanagi formula \cite{Ryu:2006bv}.  There the entanglement entropy is obtained via an analytic continuation of the Renyi entropies.  By the replica trick, the Renyi entropies in a 2D CFT are related to correlation functions of the maximal twist operators in a symmetric product orbifold of the original CFT.  It was argued in \cite{Hartman:2013mia} that the second Renyi entropy of two intervals, computed by the four-point function of twist operators, is dominated by the classical vacuum block, and it was assumed that the higher Renyi entropies behave the same.  Then by analytic continuation, the entanglement entropy is given by a classical vacuum block with parametrically small external $m_{ext}$, that is further mapped to a worldline action with two disconnected pieces.

\paragraph{Semiclassical four-point funcions in the ``light'' limit.} Four-point functions are given by sums of Virasoro blocks weighted by the structure constants squared. Consider a semiclassical four-point function with identical external operators having parametrically small weight $m_{ext}$.  In general, this four-point function receives contributions from Virasoro blocks with $m$ ranging in the entire positive real line.  However, by Proposition~\ref{Prop1}, a four-point function can always be approximated by a single block with $m\le \widehat m(m_{ext} \ll 1) \approx 1.41 \, m_{ext} < 2m_{ext}$ in the appropriate channel, and therefore always admits a worldline description in the bulk.  It would be interesting to investigate the role of the one-loop Virasoro block in this correspondence.

\paragraph{Bound states of ``light'' particles.}  Consider a CFT that is holographically dual to gravity coupled to a ``light'' particle with mass $M_{particle}$ that is of order the AdS curvature, but parametrically small.  In the CFT language, there exists a primary operator $\sigma_{particle}$ with weight $h_{particle} = m_{particle} c \ll 1$.  Proposition~\ref{Prop2b} implies that if no bound state exists with weight $m_{bound} \leq \widehat m(m_{particle})$, then the classical and one-loop structure constants in the ${\sigma_{particle} \times \sigma_{particle}}$ OPE is bounded above by \eqref{PropositionBound}, and universal \eqref{UniversalC} for $m \geq \widehat m_{\overline{vac}}(m_{particle}, 1/2)$.

\paragraph{Bulk dual of generic classical Virasoro blocks.}  One outstanding question is whether there is a similar correspondence between bulk geometry and the classical Virasoro block $f(m_1,m_2,m_3,m_4,m|x)$ with order one $m_i$ and $m$.  In this parameter regime, none of the external and internal operators should approximated as probes, and the classical Virasoro block may correspond to a classical action of a bulk geometry.  A hint of the correct bulk geometry is provided by considering the four-point function of $\bZ_2$ twist fields.  We showed in \eqref{CharacterIdentity} that the classical Virasoro block is related to the classical part of the Virasoro character.  According to \cite{Maldacena:1998bw,Maloney:2007ud}, the classical part the vacuum character is equal to the Einstein-Hilbert action plus the Gibbons-Hawking term evaluated on Euclidean AdS$_3$ with compactified time circle.

\section*{Acknowledgments}
We are grateful to Agnese Bissi, Christopher Keller, Gim Seng Ng, Shu-Heng Shao, David Simmons-Duffin, Yifan Wang, Wenbin Yan, and Xi Yin for useful discussions.  We would like to thank the 2015 BCTP Tahoe Summit and the Simons
Summer Workshop in Mathematics and Physics 2015 for their support during the course of this work.  C.M.C. is supported by BCTP Funding 39862-13070-40-PHBCTP.  Y.H.L. is supported by the Fundamental Laws Initiative Fund at Harvard University.

\appendix

\section{CFTs with a semiclassical limit}
\label{SemiclassicalLimit}

Consider a sequence of CFTs labeled by $i = 1, 2, \dotsc$, with central charges $c_i$ that are monotonically increasing and unbounded. We would like to study the behavior of this sequence of CFTs as $i$ goes to infinity. In general, it is impossible to keep track of a particular primary operator in this sequence of CFTs, as there is no canonical map from the spectrum of the $i$-th CFT to the spectrum of the $(i+1)$-th CFT. The best we can do is to consider the {\it integrated correlation functions}
\ie
\label{Integrated}
& {\cal F}^{(i)}(m_1,\cdots,m_n|x_1,\cdots,x_n)
\\
&\equiv \sum_{h_{a_1}^{(i)}\in[0,m_{1}c_i]}\sum_{h_{a_2}^{(i)}\in[0,m_{2}c_i]}\cdots\sum_{h_{a_n}^{(i)}\in[0,m_{n}c_i]}\vev{{\cal O}^{(i)}_{a_1}(x_1){\cal O}^{(i)}_{a_2}(x_2)\cdots {\cal O}^{(i)}_{a_n}(x_n)},
\fe
where ${\cal O}^{(i)}_{a}$ are primary operators  in the $i$-th CFT with weight $h^{(i)}_{a}$ that have normalized two-point functions.  This sequence is said to have a {\it semiclassical limit}, if the integrated correlation functions admit a perturbative expansion in $1/c$, in the following sense.  First we iteratively define a sequence of functions
\ie
\label{SemiClassicalAssumption}
&{\cal F}_k(m_1,\cdots,m_n|x_1,\cdots,x_n)
\\
&\equiv \lim_{i\to \infty}c_i^{k-1} \Big[ \log {\cal F}^{(i)}(m_1,\cdots,m_n|x_1,\cdots,x_n)-\sum_{m=0}^{k-1}c_i^{1-m}{\cal F}_m(m_1,\cdots,m_n|x_1,\cdots,x_n)\Big],
\fe
where the right hand side may contain logarithmic divergences independent of $x$ and $m$, that need to be properly subtracted while taking the  limit.  We demand that the limit exists and ${\cal F}_k(m_1,\cdots,m_n|x_1,\cdots,x_n)$ are continuous functions in both $m$ and $x$; furthermore their derivatives with respect to $m$ are distributions.\footnote{The definition \eqref{Integrated} of the integrated correlation functions 
is not invariant under orthogonal transformations on primary operators of the same weight.  This ambiguity may correspond to different limits \eqref{SemiClassicalAssumption}, and some of these limits may not exist.  We thank Xi Yin for pointing this out.}  Then we define the semiclassical integrated correlation functions by a formal power series
\ie\label{FormalPowerSeries}
{\cal F}(m_1,\cdots,m_n;c|x_1,\cdots,x_n) \equiv c^{\#}\exp\left(c\,{\cal F}_0 + {\cal F}_1 + c^{-1} {\cal F}_2+\cdots\right),
\fe
and the {\it semiclassical correlation function density} by taking derivatives
\ie
\label{CorrelationDensity}
F(m_1,\cdots,m_n;c|x_1,\cdots,x_n)={\partial\over \partial m_{1}}\cdots {\partial\over \partial m_{n}}F(m_{1},\cdots,m_n;c|x_1,\cdots,x_n),
\fe
which can be put into the form
\ie\label{FormalPowerSeries2}
F(m_1,\cdots,m_n;c|x_1,\cdots,x_n) = c^{\#}e^{c \, p}\left(q_0 + c^{-1}q_1+ \cdots\right),
\fe
where the $\#$'s in \eqref{FormalPowerSeries} and \eqref{FormalPowerSeries2} are $x_i$ and $m_i$ independent constants.

As an example, let us compute the two-point function density for $m \ge 1/24$.  The integrated two-point function in the $i$-th CFT is
\ie
{\cal F}^{(i)}(m_1,m_2|x_1,x_2)=\int_{c_i/24}^{{\rm min}(m_1,m_2)c_i}{\rho_i(h)\over x^{2h}}dh,
\fe
where $\rho_i(h)$ is the density of states.  By the Cardy formula \cite{Cardy:1986ie,Hartman:2014oaa} and assuming $x\le 1$, the integral is dominated in the $i\to\infty$ limit by the contribution from $m={\rm min}(m_1,m_2)$,
\ie
{\cal F}_0(m_1,m_2)=2\pi\sqrt{{1\over 6}\left({\rm min}(m_1,m_2)-{1\over 24}\right)}-2\,{\rm min}(m_1,m_2)\log x.
\fe
The semiclassical integrated two-point function is
\ie
{\cal F}(m_1,m_2|x)={1\over \sqrt{c}}\exp\Big[2\pi c\sqrt{{1\over 6}\left({\rm min}(m_1,m_2)-{1\over 24}\right)}-2c\,{\rm min}(m_1,m_2)\log x+{\cal O}(c^0)\Big],
\fe
where a logarithmic correction is also included.  The two-point function density is then given by
\ie
F(m_1,m_2|x) =  {\sqrt{c} \, \delta(m_1-m_2) \over x^{2m_1c}} {\exp\left[2\pi c\sqrt{{1\over 6}\left(m_1-{1\over 24}\right)} + {\cal O}(c^0) \right]}, \quad {\rm for}~~m_1,m_2\ge {1\over 24}.
\fe

In some special situations, we can keep track of a particular sequence of a set of operators $\{{\cal O}^{(i)}_{1},{\cal O}^{(i)}_{2},\cdots\}$, such as $\{ \sigma^n \}$ in product Ising models.  Some of their $n$-point functions may be analytically continued to the entire real line of the central charge $c_i$.  The analytically continued $n$-point function also admits a semiclassical expansion.

\section{Semiclassical Virasoro blocks}
\label{SemiclassicalCB}

In the limit of large central charge $c$ while taking the operator weights $h_i$ to scale with $c$ (fixed $m_i = h_i / c$), the Virasoro block admits a semiclassical expansion
\ie
F(h_{ext},h_{ext},h_{ext},h_{ext}, h, c | x) &= \exp\left[ - {c \over 6} f( m_{ext}, m | x) \right] g(m_{ext}, m, c | x),
\\
g(m_{ext},m,c|x)  &= \sum_{k = 0}^\infty c^{-k} {g_k(m_{ext}, m | x)}.
\fe
To the second order in the $x$-expansion,
\ie
\label{SecondOrderExpansion}
f(m_{ext}, m | x) &= 6(2m_{ext} - m) \log x - 3mx 
\\
& \hspace{.5in} - {3 (3m + 26m^2 + 16m_{ext}(m + 2m_{ext})) x^2 \over 8 (1+8m)} + {\cal O}(x^3),
\\
g_0(m_{ext},m|x)  &= 1 + {13 m x \over 2} + {(1 + 82 m + 1980 m^2 + 16224 m^3 + 43264 m^4) x^2 \over 32 (1 + 8 m)^2}
\\ 
& \quad + {16 m_{ext} (-3 + 208 m^2 + 88 m_{ext} + 4 m (5 + 104 m_{ext})) x^2 \over 32 (1 + 8 m)^2} + {\cal O}(x^3).
\fe
By computing to the sixth order in the $x$-expansion, the following properties are numerically observed to hold for fixed external weight $m_{ext} \leq 1/2$.
\begin{enumerate}
\item  $f'(m_{ext}, m | 1/2)$ is monotonically decreasing in $m$, and crosses zero only once.
\item  $f(m_{ext}, m_2 | x) - f(m_{ext}, m_1 | x)$ is monotonically decreasing in $x\in [0,1]$ for arbitrary internal weights $m_2 > m_1\geq 0$.
\item  $g_0(m_{ext},m|x) > 0$ for all internal weights $m \geq 0$ and cross ratios $0 \leq x < 1$.
\end{enumerate}

\section{The semiclassical limit of the fusion transformation}
\label{SemiclassicalFusion}

The fusion transformation relates the $s$-channel Virasoro blocks to the $t$-channel via a fusion matrix, which is defined in terms of special functions $\Gamma_b$, $S_b$, and $\Upsilon_b$ \cite{Ponsot:1999uf, Teschner:2001rv, Ponsot:2003ju}.  This appendix works out the semiclassical $b \to 0$ limit of these special functions, and computes the fusion transformation by saddle point approximation.

\subsection{The semiclassical limit of special functions}
\label{Semiclassical limit of the special functions}

The Barnes double gamma function $\Gamma_2(x|\omega_1,\omega_2)$ is defined by
\ie
\log\Gamma_2(x|\omega_1,\omega_2)={\partial\over \partial t}\sum^\infty_{n_1,n_2=0}(x+n_1\omega_1+n_2\omega_2)^{-t}\Big|_{t=0}.
\fe
The special functions $\Gamma_b(x)$, $S_b(x)$, and $\Upsilon_b(x)$ are defined by
\ie
\Gamma_b(x)={\Gamma_2(x|b,b^{-1})\over \Gamma_2(Q/2|b,b^{-1})},
\quad
S_b(x)={\Gamma_b(x)\over \Gamma_b(Q-x)},
\quad
\Upsilon_b(x)={1\over \Gamma_b(x)\Gamma_b(Q-x)},
\fe
and $\Gamma_b(x)$ function satisfies the periodic condition
\ie
\Gamma_b(x+b) = {\sqrt{2\pi} b^{bx - 1/2} \over \Gamma(bx)} \Gamma_b(x).
\fe
In the limit $b \to 0$, the periodic condition becomes a first order differential equation for $b^2 \log(\Gamma_b(y/b))$.  The solution gives the semiclassical limit of the special functions
\ie
b^2 \log\Gamma_b(y/b) &= (y - 1/2) \log \sqrt{2\pi} + {(y - 1/2)^2 \over 2} \log b - \int_{1/2}^y dz \log \Gamma(z) + {\cal O}(b),
\\
b^2 \log S_b(y/b) &= (2y - 1) \log\sqrt{2\pi} - \int_{1-y}^y dz \log \Gamma(z) + {\cal O}(b),
\\
b^2 \log \Upsilon_b(y/b) &= - (y - 1/2)^2 \log b + \int_{1/2}^y dz \log \gamma(z) + {\cal O}(b).
\fe

The expression for $S_b$ can be written in terms of polygamma functions
\ie
b^2 \log S_b(y/b) &= (2y - 1) \log\sqrt{2\pi} - \psi^{(-2)}(y) + \psi^{(-2)}(1-y) + {\cal O}(b),
\fe
where $\psi^{(-2)}(y)$ is the polygamma function of order $-2$.  It has the asymptotic behavior
\ie
\label{PolygammaAsymptotics}
- \psi^{(-2)}(i s) + \psi^{(-2)}(1-i s) =
\begin{cases}
{i \pi s^2 \over 2} - \left( {\pi \over 2} + i \log(2\pi) \right) s + {\cal O}(s^0) & s \to \infty,
\\
- {i \pi s^2 \over 2} + \left( {\pi \over 2} - i \log(2\pi) \right) s + {\cal O}(s^0) & s \to -\infty.
\end{cases}
\fe

\subsection{Correspondence between $s$- and $t$-channel blocks}
\label{Saddle correspondence across channels}

Consider the Virasoro algebra with central charge $c=1+6Q^2$, with weights parameterized by $h_\A = \A (Q-\A)$.  The $s$- and $t$- channel Virasoro blocks are related by the fusion formula \cite{Ponsot:1999uf, Teschner:2001rv, Ponsot:2003ju}
\ie
\label{FusionIntegral}
{\cal F}(h_{\A_{ext}}, h_{\A_s},c | x) = \int_{Q/2+i\bR_{\geq0}} d\A_t\, F_{\A_s \A_t}\!\! \begin{bmatrix}\A_{ext} & \A_{ext} \\ \A_{ext} & \A_{ext}\end{bmatrix}  {\cal F}(h_{\A_{ext}},h_{\A_t},c |1- x).
\fe
where for simplicity we specialize to the case $\A_1=\A_2=\A_3=\A_4=\A_{ext}$. The fusion matrix $F_{\A_s\A_t}$ is given by
\ie\label{FusionMatrix}
F_{\A_s\A_t} \!\!\begin{bmatrix}\A_{ext} & \A_{ext} \\ \A_{ext} & \A_{ext} \end{bmatrix} &= P_b(\A_s, \A_t, \A_{ext}) \times {1\over i} \int^{i\infty}_{-i\infty}ds \, T_b(\A_s, \A_t, \A_{ext}, s),
\fe
where
\ie
& P_b(\A_s, \A_t, \A_{ext}) = {\Gamma_b(2Q - 2\A_{ext} - \A_t)\Gamma_b(\A_t)^2\Gamma_b(Q-\A_t)^2\Gamma_b(Q-2\A_{ext}+\A_t)\over \Gamma_b(2Q-2\A_{ext}-\A_s)\Gamma_b(\A_s)^2\Gamma_b(Q-\A_s)^2\Gamma_b(Q-2\A_{ext}+\A_s)}
\\
& \hspace{.5in} \times {\Gamma_b(-Q+2\A_{ext}+\A_t)\Gamma_b(2\A-\A_t)\over \Gamma_b(-Q+2\A_{ext}+\A_s)\Gamma_b(2\A_{ext}-\A_s)} \times {\Gamma_b(2Q-2\A_s)\Gamma_b(2\A_s)\over \Gamma_b(Q-2\A_t)\Gamma_b(2\A_t-Q)},
\\
& T_b(\A_s, \A_t, \A, s) = {S_b(U_1+s)S_b(U_2+s)S_b(U_3+s)S_b(U_4+s)\over S_b(V_1+s)S_b(V_2+s)S_b(V_3+s)S_b(Q+s)},
\\
&U_1=\A_s,\quad U_2=Q+\A_s-2\A_{ext},\quad U_3=\A_s+2\A_{ext}-Q,\quad U_4=\A_s,
\\
&V_1=Q+\A_s-\A_t,\quad V_2=\A_s+\A_t,\quad V_3=2\A_s.
\fe
The semiclassical limit is achieved by taking $b \to 0$ while keeping $\eta_i = b\A_i$ finite.  In this limit, the Virasoro block exponentiates as
\ie
{\cal F}(h_{\A_{ext}},h_{\A},c |1- x) = \exp\left( - {1 \over b^2} f(h_{\A_{ext}} / b, h_{\A} / b, c | 1-x) + {\cal O}(1/b) \right),
\fe
and the integrals \eqref{FusionIntegral} and \eqref{FusionMatrix} can be computed by a saddle point approximation.  The terms proportional to $\log b$ all cancel in the exponent of the fusion matrix, and therefore the semiclassical limit is given by the simple replacement rule
\ie
& \Gamma_b(y/b) \to \exp\left(-{G(y) \over b^2}\right), \\
& S_b(y/b) \to \exp\left(-{G(y) - G(1-y) \over b^2}\right) = \exp\left( - {\psi^{(-2)}(y) - \psi^{(-2)}(1-y) \over b^2} \right),
\fe
where $G(y) \equiv \int_{1/2}^y dz \log\Gamma(z)$, and $\psi^{(-2)}(y)$ is the polygamma function of order $-2$.  Note that $G(y)$ is not a meromorphic function.

Let us define $s = i \sigma / b$.  The $\sigma$ integral in \eqref{FusionMatrix} is dominated by points maximizing the real part of $\log T_b$ along the integration contour.  In the following, we will assume $\eta_s + 2 \eta > 1$; otherwise infinitely many poles will cross the $\sigma$ integral contour as we take $b \to 0$.

\paragraph{Real $\bf \eta_s$ ($\bf h_s \leq c/24$).}  When $\eta_s$ is real, the $\sigma$ integral is dominated by the contribution at $\sigma = 0$
\ie
& \psi^{(-2)}(\eta_s+\eta_t)  - \psi^{(-2)}(1-\eta_s-\eta_t) - \psi^{(-2)}(-\eta_s+\eta_t) + \psi^{(-2)}(1+\eta_s-\eta_t) + \dotsb
\\
&= G(\eta_s + \eta_t) - G(1 - \eta_s - \eta_t) - G(- \eta_s + \eta_t) + G(1 + \eta_s - \eta_t) + \dotsb,
\fe
where the omitted terms $\dotsb$ do not involve $\eta_t$.  Note that the $\eta_t$-dependent factors cancel for when the $s$-channel is the vacuum $\eta_s = 0$.  The other $\eta_t$-dependent factors in the fusion matrix can be written as
\ie
& \lim_{b \to 0} b^2 \log P_b(\eta_s/b, \eta_t/b, \eta_{ext}/b) 
\\
& = - 2 H(\eta_t) - H(2\eta_{ext} + \eta_t - 1) - H(2\eta_{ext} - \eta_t) + G(1-2\eta_t) + G(2\eta_t -1),
\fe
where $H(y) \equiv G(y) + G(1-y) = \int_{1/2}^y dz \log \gamma(z)$.  Thus for real $\eta_s$, the $s$-channel block is equal to the fusion matrix times the $t$-channel block evaluated at the solution to 
\ie
\label{SaddleCorrespondenceEq}
0 &= - 2 \log\gamma(\eta_t) - \log\gamma(2\eta_{ext}+\eta_t-1) + \log\gamma(2\eta_{ext}-\eta_t) - 2\log\Gamma(1-2\eta_t) + 2\log\Gamma(2\eta_t-1)
\\
& \hspace{.5in} + \log\Gamma(\eta_s+\eta_t) + \log\Gamma(1-\eta_s-\eta_t) - \log\Gamma(-\eta_s+\eta_t) - \log\Gamma(1+\eta_s-\eta_t)
\\
& \hspace{1in} - {d \over d\eta_t} f(\eta_{ext}, \eta_t | 1-x).
\fe

\paragraph{Complex $\eta_s$ ($\bf h_s \geq c/24$).}  When $\eta_s \in 1/2 + i\bR$, the $\sigma$ integral has maximal real part of the exponent on the whole segment $- 2\,\text{Im}\,\eta_s \leq \sigma \leq 0$, where the $\eta_t$-dependent piece is
\ie
\left( {1 \over 2} - \eta_t \right) \left[ ( \log(1-2\eta_t) - \log(2\eta_t-1) \right].
\fe
Note that the $\eta_s$ dependence is gone.  In the fusion transformation, the $s$-channel block is dominated by the $t$-channel at the solution to
\ie
\label{SaddleCorrespondenceComplex}
0 &= - 2 \log\gamma(\eta_t) - \log\gamma(2\eta_{ext}+\eta_t-1) + \log\gamma(2\eta_{ext}-\eta_t)+ \log\gamma(2\eta_t-1) + \log\gamma(2\eta_t)
\\
& \hspace{0.5in} - {d \over d\eta_t} f(\eta_{ext}, \eta_t | 1-x).
\fe

\bibliography{refs} 

\providecommand{\href}[2]{#2}\begingroup\raggedright\begin{thebibliography}{10}

\bibitem{Polyakov:1974gs}
A.~M. Polyakov, {\it {Nonhamiltonian approach to conformal quantum field
  theory}},  {\em Zh. Eksp. Teor. Fiz.} {\bf 66} (1974) 23--42.

\bibitem{Ferrara:1973yt}
S.~Ferrara, A.~F. Grillo, and R.~Gatto, {\it {Tensor representations of
  conformal algebra and conformally covariant operator product expansion}},
  {\em Annals Phys.} {\bf 76} (1973) 161--188.

\bibitem{Mack:1975jr}
G.~Mack, {\it {Duality in quantum field theory}},  {\em Nucl. Phys.} {\bf B118}
  (1977) 445--457.

\bibitem{Belavin:1984vu}
A.~A. Belavin, A.~M. Polyakov, and A.~B. Zamolodchikov, {\it {Infinite
  Conformal Symmetry in Two-Dimensional Quantum Field Theory}},  {\em Nucl.
  Phys.} {\bf B241} (1984) 333--380.

\bibitem{Knizhnik:1984nr}
V.~G. Knizhnik and A.~B. Zamolodchikov, {\it {Current Algebra and Wess-Zumino
  Model in Two-Dimensions}},  {\em Nucl. Phys.} {\bf B247} (1984) 83--103.

\bibitem{Gepner:1986wi}
D.~Gepner and E.~Witten, {\it {String Theory on Group Manifolds}},  {\em Nucl.
  Phys.} {\bf B278} (1986) 493.

\bibitem{Bouwknegt:1992wg}
P.~Bouwknegt and K.~Schoutens, {\it {W symmetry in conformal field theory}},
  {\em Phys. Rept.} {\bf 223} (1993) 183--276,
  [\href{http://xxx.lanl.gov/abs/hep-th/9210010}{{\tt hep-th/9210010}}].

\bibitem{Verlinde:1988sn}
E.~P. Verlinde, {\it {Fusion Rules and Modular Transformations in 2D Conformal
  Field Theory}},  {\em Nucl. Phys.} {\bf B300} (1988) 360.

\bibitem{Dijkgraaf:1988tf}
R.~Dijkgraaf and E.~P. Verlinde, {\it {Modular Invariance and the Fusion
  Algebra}},  {\em Nucl. Phys. Proc. Suppl.} {\bf 5} (1988) 87--97.

\bibitem{Moore:1988uz}
G.~W. Moore and N.~Seiberg, {\it {Polynomial Equations for Rational Conformal
  Field Theories}},  {\em Phys. Lett.} {\bf B212} (1988) 451.

\bibitem{Moore:1988ss}
G.~W. Moore and N.~Seiberg, {\it {Naturality in Conformal Field Theory}},  {\em
  Nucl. Phys.} {\bf B313} (1989) 16.

\bibitem{Rattazzi:2008pe}
R.~Rattazzi, V.~S. Rychkov, E.~Tonni, and A.~Vichi, {\it {Bounding scalar
  operator dimensions in 4D CFT}},  {\em JHEP} {\bf 12} (2008) 031,
  [\href{http://xxx.lanl.gov/abs/0807.0004}{{\tt arXiv:0807.0004}}].

\bibitem{ElShowk:2012ht}
S.~El-Showk, M.~F. Paulos, D.~Poland, S.~Rychkov, D.~Simmons-Duffin, and
  A.~Vichi, {\it {Solving the 3D Ising Model with the Conformal Bootstrap}},
  {\em Phys. Rev.} {\bf D86} (2012) 025022,
  [\href{http://xxx.lanl.gov/abs/1203.6064}{{\tt arXiv:1203.6064}}].

\bibitem{Fitzpatrick:2012yx}
A.~L. Fitzpatrick, J.~Kaplan, D.~Poland, and D.~Simmons-Duffin, {\it {The
  Analytic Bootstrap and AdS Superhorizon Locality}},  {\em JHEP} {\bf 12}
  (2013) 004, [\href{http://xxx.lanl.gov/abs/1212.3616}{{\tt
  arXiv:1212.3616}}].

\bibitem{Komargodski:2012ek}
Z.~Komargodski and A.~Zhiboedov, {\it {Convexity and Liberation at Large
  Spin}},  {\em JHEP} {\bf 11} (2013) 140,
  [\href{http://xxx.lanl.gov/abs/1212.4103}{{\tt arXiv:1212.4103}}].

\bibitem{Alday:2015eya}
L.~F. Alday, A.~Bissi, and T.~Lukowski, {\it {Large spin systematics in CFT}},
  \href{http://xxx.lanl.gov/abs/1502.0770}{{\tt arXiv:1502.0770}}.

\bibitem{Ooguri:2015Talk}
H.~Ooguri, {\it Analytic bootstrap bounds},  {\em
  https://www.perimeterinstitute.ca/videos/analytic-bootstrap-bounds} (2015).

\bibitem{Maldacena:1997re}
J.~M. Maldacena, {\it {The Large N limit of superconformal field theories and
  supergravity}},  {\em Int. J. Theor. Phys.} {\bf 38} (1999) 1113--1133,
  [\href{http://xxx.lanl.gov/abs/hep-th/9711200}{{\tt hep-th/9711200}}]. [Adv.
  Theor. Math. Phys.2,231(1998)].

\bibitem{Gubser:1998bc}
S.~S. Gubser, I.~R. Klebanov, and A.~M. Polyakov, {\it {Gauge theory
  correlators from noncritical string theory}},  {\em Phys. Lett.} {\bf B428}
  (1998) 105--114, [\href{http://xxx.lanl.gov/abs/hep-th/9802109}{{\tt
  hep-th/9802109}}].

\bibitem{Witten:1998qj}
E.~Witten, {\it {Anti-de Sitter space and holography}},  {\em Adv. Theor. Math.
  Phys.} {\bf 2} (1998) 253--291,
  [\href{http://xxx.lanl.gov/abs/hep-th/9802150}{{\tt hep-th/9802150}}].

\bibitem{Brown:1986nw}
J.~D. Brown and M.~Henneaux, {\it {Central Charges in the Canonical Realization
  of Asymptotic Symmetries: An Example from Three-Dimensional Gravity}},  {\em
  Commun. Math. Phys.} {\bf 104} (1986) 207--226.

\bibitem{Banados:1992wn}
M.~Banados, C.~Teitelboim, and J.~Zanelli, {\it {The Black hole in
  three-dimensional space-time}},  {\em Phys. Rev. Lett.} {\bf 69} (1992)
  1849--1851, [\href{http://xxx.lanl.gov/abs/hep-th/9204099}{{\tt
  hep-th/9204099}}].

\bibitem{Zamolodchikov:1985ie}
A.~Zamolodchikov, {\it {CONFORMAL SYMMETRY IN TWO-DIMENSIONS: AN EXPLICIT
  RECURRENCE FORMULA FOR THE CONFORMAL PARTIAL WAVE AMPLITUDE}},  {\em
  Commun.Math.Phys.} {\bf 96} (1984) 419--422.

\bibitem{Hellerman:2009bu}
S.~Hellerman, {\it {A Universal Inequality for CFT and Quantum Gravity}},  {\em
  JHEP} {\bf 08} (2011) 130, [\href{http://xxx.lanl.gov/abs/0902.2790}{{\tt
  arXiv:0902.2790}}].

\bibitem{Ponsot:1999uf}
B.~Ponsot and J.~Teschner, {\it {Liouville bootstrap via harmonic analysis on a
  noncompact quantum group}},
  \href{http://xxx.lanl.gov/abs/hep-th/9911110}{{\tt hep-th/9911110}}.

\bibitem{Teschner:2001rv}
J.~Teschner, {\it {Liouville theory revisited}},  {\em Class. Quant. Grav.}
  {\bf 18} (2001) R153--R222,
  [\href{http://xxx.lanl.gov/abs/hep-th/0104158}{{\tt hep-th/0104158}}].

\bibitem{Ponsot:2003ju}
B.~Ponsot, {\it {Recent progresses on Liouville field theory}},  {\em Int. J.
  Mod. Phys.} {\bf A19S2} (2004) 311--335,
  [\href{http://xxx.lanl.gov/abs/hep-th/0301193}{{\tt hep-th/0301193}}].

\bibitem{Hartman:2013mia}
T.~Hartman, {\it {Entanglement Entropy at Large Central Charge}},
  \href{http://xxx.lanl.gov/abs/1303.6955}{{\tt arXiv:1303.6955}}.

\bibitem{Maldacena:2015iua}
J.~Maldacena, D.~Simmons-Duffin, and A.~Zhiboedov, {\it {Looking for a bulk
  point}},  \href{http://xxx.lanl.gov/abs/1509.0361}{{\tt arXiv:1509.0361}}.

\bibitem{Dijkgraaf:1996xw}
R.~Dijkgraaf, G.~W. Moore, E.~P. Verlinde, and H.~L. Verlinde, {\it {Elliptic
  genera of symmetric products and second quantized strings}},  {\em Commun.
  Math. Phys.} {\bf 185} (1997) 197--209,
  [\href{http://xxx.lanl.gov/abs/hep-th/9608096}{{\tt hep-th/9608096}}].

\bibitem{Yin:2007gv}
X.~Yin, {\it {Partition Functions of Three-Dimensional Pure Gravity}},  {\em
  Commun. Num. Theor. Phys.} {\bf 2} (2008) 285--324,
  [\href{http://xxx.lanl.gov/abs/0710.2129}{{\tt arXiv:0710.2129}}].

\bibitem{Headrick:2010zt}
M.~Headrick, {\it {Entanglement Renyi entropies in holographic theories}},
  {\em Phys. Rev.} {\bf D82} (2010) 126010,
  [\href{http://xxx.lanl.gov/abs/1006.0047}{{\tt arXiv:1006.0047}}].

\bibitem{Cardy:1986ie}
J.~L. Cardy, {\it {Operator Content of Two-Dimensional Conformally Invariant
  Theories}},  {\em Nucl. Phys.} {\bf B270} (1986) 186--204.

\bibitem{Hartman:2014oaa}
T.~Hartman, C.~A. Keller, and B.~Stoica, {\it {Universal Spectrum of 2d
  Conformal Field Theory in the Large c Limit}},  {\em JHEP} {\bf 09} (2014)
  118, [\href{http://xxx.lanl.gov/abs/1405.5137}{{\tt arXiv:1405.5137}}].

\bibitem{Carlip:2000nv}
S.~Carlip, {\it {Logarithmic corrections to black hole entropy from the Cardy
  formula}},  {\em Class. Quant. Grav.} {\bf 17} (2000) 4175--4186,
  [\href{http://xxx.lanl.gov/abs/gr-qc/0005017}{{\tt gr-qc/0005017}}].

\bibitem{Nakayama:2004vk}
Y.~Nakayama, {\it {Liouville field theory: A Decade after the revolution}},
  {\em Int. J. Mod. Phys.} {\bf A19} (2004) 2771--2930,
  [\href{http://xxx.lanl.gov/abs/hep-th/0402009}{{\tt hep-th/0402009}}].

\bibitem{Dorn:1992at}
H.~Dorn and H.~J. Otto, {\it {On correlation functions for noncritical strings
  with $c <= 1$ $d >= 1$}},  {\em Phys. Lett.} {\bf B291} (1992) 39--43,
  [\href{http://xxx.lanl.gov/abs/hep-th/9206053}{{\tt hep-th/9206053}}].

\bibitem{Dorn:1994xn}
H.~Dorn and H.~J. Otto, {\it {Two and three point functions in Liouville
  theory}},  {\em Nucl. Phys.} {\bf B429} (1994) 375--388,
  [\href{http://xxx.lanl.gov/abs/hep-th/9403141}{{\tt hep-th/9403141}}].

\bibitem{Zamolodchikov:1995aa}
A.~B. Zamolodchikov and A.~B. Zamolodchikov, {\it {Structure constants and
  conformal bootstrap in Liouville field theory}},  {\em Nucl. Phys.} {\bf
  B477} (1996) 577--605, [\href{http://xxx.lanl.gov/abs/hep-th/9506136}{{\tt
  hep-th/9506136}}].

\bibitem{Teschner:1995yf}
J.~Teschner, {\it {On the Liouville three point function}},  {\em Phys. Lett.}
  {\bf B363} (1995) 65--70, [\href{http://xxx.lanl.gov/abs/hep-th/9507109}{{\tt
  hep-th/9507109}}].

\bibitem{Teschner:2001hk}
J.~Teschner, {\it {Consistency of the bootstrap for Liouville field theory}},
  in {\em {Low dimensional integrable models and their applications in field
  theory and statistical physics. Proceedings, Workshop, Annecy, France, May
  21-25, 2001}}, pp.~81--91, 2001.

\bibitem{Hadasz:aa}
L.~Hadasz, Z.~Jaskolski, and M.~Piatek, {\it Classical geometry from the
  quantum liouville theory},
  \href{http://xxx.lanl.gov/abs/hep-th/0504204}{{\tt hep-th/0504204}}.

\bibitem{Zamolodchikov:1987}
A.~Zamolodchikov, {\it {Conformal symmetry in two-dimensional space: Recursion
  representation of conformal block}},  {\em Theoretical and Mathematical
  Physics} {\bf 73} (1987), no.~1 1088--1093.

\bibitem{Lunin:2000yv}
O.~Lunin and S.~D. Mathur, {\it {Correlation functions for M**N / S(N)
  orbifolds}},  {\em Commun. Math. Phys.} {\bf 219} (2001) 399--442,
  [\href{http://xxx.lanl.gov/abs/hep-th/0006196}{{\tt hep-th/0006196}}].

\bibitem{Headrick:2015gba}
M.~Headrick, A.~Maloney, E.~Perlmutter, and I.~G. Zadeh, {\it {R{\'e}nyi
  entropies, the analytic bootstrap, and 3D quantum gravity at higher genus}},
  {\em JHEP} {\bf 07} (2015) 059,
  [\href{http://xxx.lanl.gov/abs/1503.0711}{{\tt arXiv:1503.0711}}].

\bibitem{Bouwknegt:1988sv}
P.~Bouwknegt, {\it {EXTENDED CONFORMAL ALGEBRAS}},  {\em Phys. Lett.} {\bf
  B207} (1988) 295.

\bibitem{Keller:2013qqa}
C.~A. Keller, {\it {Modularity, Calabi-Yau geometry and 2d CFTs}},  {\em Proc.
  Symp. Pure Math.} {\bf 88} (2014) 307--316,
  [\href{http://xxx.lanl.gov/abs/1312.7313}{{\tt arXiv:1312.7313}}].

\bibitem{Witten:2007kt}
E.~Witten, {\it {Three-Dimensional Gravity Revisited}},
  \href{http://xxx.lanl.gov/abs/0706.3359}{{\tt arXiv:0706.3359}}.

\bibitem{frenkel1985moonshine}
I.~B. Frenkel, J.~Lepowsky, and A.~Meurman, {\it A moonshine module for the
  monster},  in {\em Vertex operators in mathematics and physics},
  pp.~231--273.
\newblock Springer, 1985.

\bibitem{Fitzpatrick:2014vua}
A.~L. Fitzpatrick, J.~Kaplan, and M.~T. Walters, {\it {Universality of
  Long-Distance AdS Physics from the CFT Bootstrap}},  {\em JHEP} {\bf 08}
  (2014) 145, [\href{http://xxx.lanl.gov/abs/1403.6829}{{\tt
  arXiv:1403.6829}}].

\bibitem{Hijano:2015rla}
E.~Hijano, P.~Kraus, and R.~Snively, {\it {Worldline approach to semi-classical
  conformal blocks}},  {\em JHEP} {\bf 07} (2015) 131,
  [\href{http://xxx.lanl.gov/abs/1501.0226}{{\tt arXiv:1501.0226}}].

\bibitem{Hijano:2015qja}
E.~Hijano, P.~Kraus, E.~Perlmutter, and R.~Snively, {\it {Semiclassical
  Virasoro Blocks from AdS$_3$ Gravity}},
  \href{http://xxx.lanl.gov/abs/1508.0498}{{\tt arXiv:1508.0498}}.

\bibitem{Ryu:2006bv}
S.~Ryu and T.~Takayanagi, {\it {Holographic derivation of entanglement entropy
  from AdS/CFT}},  {\em Phys. Rev. Lett.} {\bf 96} (2006) 181602,
  [\href{http://xxx.lanl.gov/abs/hep-th/0603001}{{\tt hep-th/0603001}}].

\bibitem{Maldacena:1998bw}
J.~M. Maldacena and A.~Strominger, {\it {AdS(3) black holes and a stringy
  exclusion principle}},  {\em JHEP} {\bf 12} (1998) 005,
  [\href{http://xxx.lanl.gov/abs/hep-th/9804085}{{\tt hep-th/9804085}}].

\bibitem{Maloney:2007ud}
A.~Maloney and E.~Witten, {\it {Quantum Gravity Partition Functions in Three
  Dimensions}},  {\em JHEP} {\bf 02} (2010) 029,
  [\href{http://xxx.lanl.gov/abs/0712.0155}{{\tt arXiv:0712.0155}}].

\end{thebibliography}\endgroup
\bibliographystyle{JHEP}

\end{document}